\documentclass[notitlepage,a4paper,superscriptaddress,10pt,tightenlines,twocolumn,pra,nofootinbib]{revtex4-2}
\usepackage[utf8]{inputenc}
\usepackage{graphicx}
\usepackage{physics, amsmath, amssymb, amsthm, units, dsfont, bm}
\usepackage{orcidlink}
\usepackage{mathtools, amsfonts, mathrsfs, bbm}
\usepackage{tikz}

\usepackage[dvipsnames]{xcolor}
\usepackage{adjustbox}

\usepackage{tikz}
\usepackage{tikzsymbols}
\usetikzlibrary{math,calc}
\usetikzlibrary{overlay-beamer-styles}
\usetikzlibrary{arrows.meta}
\usetikzlibrary{automata, positioning, arrows}
\usetikzlibrary{backgrounds}
\usetikzlibrary{decorations.markings}

\colorlet{alice}{Green}
\colorlet{bob}{Blue}
\colorlet{charlie}{Red}

\newcommand{\kommentar}[1]{}

\newcommand{\id}{\mathds{1}}

\renewcommand{\tr}{\mathrm{tr}}

\renewcommand{\norm}[1]{\left\lVert#1\right\rVert}

\renewcommand{\ket}[1]{| #1 \rangle}
\renewcommand{\bra}[1]{\langle #1 |}
\renewcommand{\ketbra}[2]{| #1 \rangle \! \langle #2 |}
\renewcommand{\braket}[2]{\langle #1 | #2 \rangle}

\renewcommand{\varrho}{\rho}

\newtheorem{theorem}{Theorem}

\newtheorem{observation}[theorem]{Observation}

\makeatletter
\def\maketitle{
\@author@finish
\title@column\titleblock@produce
\suppressfloats[t]}
\makeatother

\begin{document}


\title{Complete Hierarchies for the Geometric Measure of Entanglement}

\author{Lisa T. Weinbrenner$^{\orcidlink{0009-0001-0598-0779}}$
}
\affiliation{Naturwissenschaftlich-Technische Fakult\"{a}t, Universit\"{a}t Siegen, Walter-Flex-Stra\ss e 3, 57068 Siegen, Germany}

\author{
Albert Rico$^{\orcidlink{0000-0001-8211-499X}}$
}
\affiliation{Naturwissenschaftlich-Technische Fakult\"{a}t, Universit\"{a}t Siegen, Walter-Flex-Stra\ss e 3, 57068 Siegen, Germany}
\affiliation{GIQ - Quantum Information Group, Department of Physics, Autonomous University of Barcelona, Spain}

\author{Kenneth Goodenough$^{\orcidlink{0000-0002-1761-0038}}$
}
\affiliation{Naturwissenschaftlich-Technische Fakult\"{a}t, Universit\"{a}t Siegen, Walter-Flex-Stra\ss e 3, 57068 Siegen, Germany}

\author{
Xiao-Dong Yu$^{\orcidlink{0000-0001-8835-5524}}$
}
\affiliation{Department of Physics, Shandong University, Jinan 250100, China}

\author{Otfried Gühne$^{\orcidlink{0000-0002-6033-0867}}$}
\affiliation{Naturwissenschaftlich-Technische Fakult\"{a}t, Universit\"{a}t Siegen, Walter-Flex-Stra\ss e 3, 57068 Siegen, Germany}

\date{\today}

\begin{abstract}
In quantum physics, multiparticle systems are described by 
quantum states acting on tensor products of Hilbert spaces. 
This product structure leads to the distinction between product 
states and entangled states; moreover, one can quantify entanglement 
by considering the distance of a quantum state to the set of 
product states. The underlying optimization problem occurs
frequently in physics and beyond, for instance in the 
computation of the injective tensor norm in multilinear algebra.
Here, we introduce a method to determine 
the maximal overlap of a pure multiparticle quantum state with 
product states based on considering several copies of the 
pure state. This leads to three types of hierarchical
approximations to the problem, all of which we prove to converge
to the actual value. Besides allowing for the computation of the 
geometric measure of entanglement, our results can be used 
to tackle optimizations over stochastic local transformations, to find 
entanglement witnesses for weakly entangled bipartite states, and to
design strong separability tests for mixed multiparticle states. Finally, our approach sheds  light on the complexity of separability tests.
\end{abstract}

\maketitle

\section{Introduction}
Entanglement is one of the most studied resources in quantum information theory, being useful for quantum metrology, quantum cryptography or quantum communication~\cite{horodecki2009quantum, guhne2009entanglement}. As such, there are numerous ways to characterize and quantify entanglement from different viewpoints. Especially for bipartite pure states the problem of quantifying entanglement is well understood. However, for most practical applications multiparticle states are needed, calling for a better understanding of 
multiparticle entanglement quantification.

One possible measure of entanglement for quantum states 
is given by the geometric measure of entanglement, which quantifies the entanglement of a state by its geometric distance to separable states~\cite{shimony1995degree, barnum2001monotones, wei2003geometric, weinbrenner2025quantifying}. For a pure three-particle state $\ket{\psi}$ it is defined as $E_G(\psi) = 1-\Lambda^2(\psi)$, where 
\begin{align}
    \Lambda^2(\psi) = \max_{\ket{abc}} |\braket{abc}{\psi}|^2 
    \label{eq:defGeometricMeasure}
\end{align}
denotes the maximal overlap between the three-partite $\ket{\psi}$ and all product states $\ket{abc}$. 
Interestingly, this measure is directly related to different notions considered in mathematics like the injective tensor norm 
and tensor eigenvalues~\cite{friedland2025tensors, bruzda2024rank}. For bipartite pure states the measure can be directly evaluated by the Schmidt decomposition; moreover, there 
are results 
for 
special families of multiparticle states \cite{weinbrenner2025quantifying}.
However, the calculation of the geometric measure is known to be computationally difficult for general multipartite states~\cite{harrow2013testing}, implying a similar statement for the injective tensor norm.

There are several examples of hard optimization problems 
which can be tackled with so-called hierarchies, forming 
an infinite sequence of approximations indexed by the level $k$. Larger $k$ yield tighter bounds, at the cost of increased computational complexity. A hierarchy is complete if, informally speaking, the approximations converge to the actual solution as $k$ tends to infinity. Some notable examples in quantum physics include hierarchies for the separability problem \cite{Doherty2004complete}, the marginal problem \cite{yu2021complete}, 
and the question whether or not observed correlations can originate from quantum
mechanics \cite{navascues2008convergent}, and in mathematics the optimization of multivariate polynomials \cite{lasserre2001global}.

In this paper we describe three different hierarchies for the geometric measure of pure states which are provably complete. All three hierarchies rely on a multi-copy approach, where given operators act on multiple copies of the target state, leading to increasingly better bounds with the number of copies considered. We illustrate numerically the convergence behavior of the hierarchies for some three- and five-qubit states. Furthermore, we show how the results can be adapted for witness constructions and checking distillability, and use our approach
to characterize entanglement in mixed states. Interestingly, one of the hierarchies is closely related to so-called product tests for entanglement \cite{mintert2005concurrence,beausoleil2006test,harrow2013testing,foulds2021controlled}, leading to a natural improvement of this test to more copies.
From a mathematical perspective, our work answers a question posed by Harald~A.~Helfgott \cite{helfgott}, when looking to generalize a well known 
upper bound on the injective norm  (or the largest eigenvalue) of real 
symmetric matrices to general real symmetric tensors. This was taken up
by Shmuel Friedland by introducing hierarchies for symmetric tensors \cite{friedland2021upperboundsspectralnorm}. Our formulations are an extension of these in a different language without requiring symmetry assumptions; moreover, they are complete.

\section{The Hierarchies}
Let us start by describing our first hierarchy $\mathfrak{H}_1$ 
in the basic setting of three particles of the same dimension; the same approach generalizes in a straightforward manner to any number of particles and dimensions. For that, note that 
$\Lambda^2$ can be written on a two-copy
space as
\begin{align}\label{eq:genidea}
    \Lambda^2(\psi) &= \max_{\ket{abc}} |\bra{abc}^{\otimes 2}\ket{\psi}^{\otimes 2}| \notag \\ 
    &= \max_{\ket{abc}} |\bra{abc}^{\otimes 2} \Pi_2\! \otimes \Pi_2\! \otimes \Pi_2 \ket{\psi}^{\otimes 2}|.
\end{align}
Clearly, the symmetry of Alice's state allows her to project her state on the symmetric subspace of two particles without changing the value of the above scalar product, i.e., she may apply the projection onto the symmetric space $\Pi_2=(\openone+V)/2$
with the SWAP operator $V$ acting as $V\ket{\psi}\ket{\phi}=\ket{\phi}\ket{\psi}$.
The main point is now that while $\Pi_2$ acts trivially on the two copies of $\ket{a}$, it acts nontrivially on the two copies of $\ket{\psi}$. The maximum over product states can therefore be relaxed to the norm of the resulting vector,
\begin{equation}
    \Lambda^2 \leq \Vert \ket{F_2} \Vert
    \mbox{ with }
    \ket{F_2} = \Pi_2^{\otimes 3} [\ket{\psi}^{\otimes 2}],
    \label{eq:genidea2}
\end{equation}
where we have used $\Pi_2^{\otimes 3}$ to denote $\Pi_2\otimes \Pi_2\otimes \Pi_2$ for simplicity. Note that in general the dimensions of the three subsystems do not need to be identical.
This approach generalizes naturally to more copies, where the application of the symmetric projector $\Pi_k$ on $k$ copies of $\ket{\psi}$ leads to the upper bound 
$\Lambda^2 \leq \Vert \ket{F_k} \Vert^{2/k}$.

The relevant question is whether these upper bounds converge
to the correct value of the injective tensor norm $\Lambda^2.$
In fact, we show in Appendix~\ref{app:proof_convergence_first_hierarchy} that the hierarchy is complete from both sides; it delivers converging upper {\it and} lower bounds. We can directly formulate our first main result:

\noindent
\begin{observation}\label{obs:convergence}
For any $k$, we have the two-sided bound
\begin{equation}
    d_k \Vert \ket{F_k} \Vert^{2/k}
    \leq 
    \Lambda^2(\psi)
    \leq
    \Vert \ket{F_k} \Vert^{2/k},
    \label{eq-observation1}
\end{equation}
where the coefficients $d_k<1$ are explicitly known and converge
to one, $\lim_{k\rightarrow\infty} d_k =1$.
\end{observation}

In the proof of this observation (given in
Appendix~\ref{app:proof_convergence_first_hierarchy})
one expresses the symmetric projectors $\Pi_k$ by integrals over 
symmetric pure product states, which allows us to estimate the 
norm $\Vert \ket{F_k} \Vert$ by powers of $\Lambda$ up to some 
constants incorporated in $d_k$. This points at a physical 
interpretation coming from the quantum de Finetti theorem~\cite{renner2007symmetry, caves2002unknown,christandl2007one},
which roughly speaking states that a symmetric $k$-particle quantum state $\varrho$ cannot be far away from a product state, that is $\varrho \approx \sigma^{\otimes k}$. 
An important point for our results later is that Eq.~\eqref{eq-observation1} 
delivers converging upper and lower bounds in terms of analytical multi-copy expressions. General lower bounds can be obtained with various methods \cite{weinbrenner2025quantifying}, and discrete upper and lower bounds
can be obtained by a brute-force optimization over discretizations of the product
states \cite{friedland2025tensors}.

The attentive reader will have realized that the estimate used in the transition from Eq.~(\ref{eq:genidea}) to Eq.~(\ref{eq:genidea2}) is not optimal: The maximization over 
six-fold product states on the left-hand side of the scalar product was replaced by a maximization over general states, neglecting any product structure. A more careful estimate 
may keep at least the product structure for a split of the six particles into two groups. Indeed, for such bipartitions, the maximal overlap with $\ket{F_2}$ can directly 
be computed by considering the Schmidt decomposition of $\ket{F_2}$ 
\cite{wei2003geometric, bourennane2004experimental}. This gives, with small extra 
effort, an improved bound on $\Lambda^2$ for a fixed level of a hierarchy, see also Appendix~\ref{subsec:Herarchy1}.


For formulating the second hierarchy $\mathfrak{H}_2$ one uses a specific property of the maximization in the definition of $\Lambda^2$. Namely, if 
two vectors (say, $\ket{b}$ and $\ket{c}$) are fixed, then 
$\braket{bc}{\psi} = \ket{\alpha}$ is an unnormalized pure state and the optimal 
$\ket{a}$ needs to be proportional to $\ket{\alpha}$. This implies that
taking the optimal $\ket{b_o c_o}$ enforces 
$\ket{\alpha} = \Lambda \ket{a_o}$ (up to some phase), leading to
$
\braket{\psi}{b_o c_o}\braket{\psi}{a_o c_o}\braket{\psi}{a_o b_o} 
= \Lambda^3(\psi) \bra{a_o b_o c_o}
$
and therefore 
\begin{align}\label{eq:2nd_hierarchy_example}
    \Lambda^4(\psi) = \max_{\ket{abc}} |\braket{\psi}{b c}\braket{\psi}{a c}\braket{\psi}{a b}\ket{\psi}|.
\end{align}
Similar to 
Eq.~\eqref{eq:genidea} one can insert now projectors onto symmetric 
spaces on the two copies of $\ket{a}$ and similarly for the two
copies of $\ket{b}$ and $\ket{c}$. Then, one relaxes the optimization
by taking the norm 
of the resulting six-particle vector $\ket{G_4}.$ This kind of estimate can be 
generalized to more copies of $\ket{\psi}$ in different ways, since there are 
several ways to wire the indices of the tensors corresponding to the states 
$\psi$, see Fig.~\ref{fig:second_hierarchy_wirings} for examples. All possible tree graphs (that is, connected graphs without cycles) lead to a valid 
multi-copy tensor which gives an upper bound on $\Lambda$. We found that 
in practice path graphs of length $k$ worked well, and will thus use 
those for the $k$'th level in the hierarchy, see Appendix~\ref{app:trees} for 
more information. The bounds from this approach are tighter than the 
ones from the first hierarchy since, for a given number of copies of $\ket{\psi}$, the optimization over fewer external indices needs to be relaxed.
Analogous completeness results
as Observation 1 hold here as well; we refer to Appendix~\ref{app:proof_convergence_second_hierarchy} 
for a detailed discussion. Furthermore, in Appendix~\ref{app:connection_helfgott_friedland}
we explain the relation of our hierarchies $\mathfrak{H}_1$ and
$\mathfrak{H}_2$ to results known in mathematics, and demonstrate how the hierarchies can also be adapted for the optimization over real instead of complex product vectors.

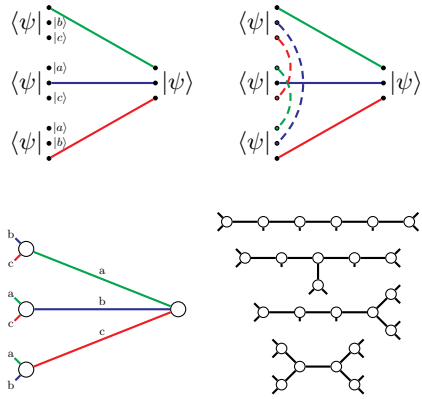
\begin{figure}[t]
    \centering
    \begin{minipage}{\linewidth}
        \centering
        \begin{tikzpicture}

	\tikzmath{
	\x=2; 
	\y=0.8; 
	\lx=0.3; 
	\ly=0.2; 
	\r=0.5pt; 
	\la=0.15; 
	\s=0.5; 
	\w=0.2; 
	} 
	
\begin{scope}[scale=1]

	\node[draw=none](C1) at (0,0) {$\langle \psi |$};
	\node[draw=none] (C2) at (0,-\y) {$\langle \psi |$};
	\node[draw=none] (C3) at (0,-2*\y) {$\langle \psi |$};
	
	\node[draw=none] (C4) at (\x,-\y) {$| \psi \rangle$};
	
	\node[fill=black,circle,inner sep=\r,draw] (C11) at (\lx,0+\ly) {};
	\node[fill=black,circle,inner sep=\r,draw] (C12) at (\lx,0) {};
	\node[fill=black,circle,inner sep=\r,draw] (C13) at (\lx,0-\ly) {};
	
	\node[fill=black,circle,inner sep=\r,draw] (C21) at (\lx,-\y+\ly) {};
	\node[fill=black,circle,inner sep=\r,draw] (C22) at (\lx,-\y) {};
	\node[fill=black,circle,inner sep=\r,draw] (C23) at (\lx,-\y-\ly) {};
	
	\node[fill=black,circle,inner sep=\r,draw] (C31) at (\lx,-2*\y+\ly) {};
	\node[fill=black,circle,inner sep=\r,draw] (C32) at (\lx,-2*\y) {};
	\node[fill=black,circle,inner sep=\r,draw] (C33) at (\lx,-2*\y-\ly) {};
	
	\node[fill=black,circle,inner sep=\r,draw] (C41) at (\x-\lx,-\y+\ly) {};
	\node[fill=black,circle,inner sep=\r,draw] (C42) at (\x-\lx,-\y) {};
	\node[fill=black,circle,inner sep=\r,draw] (C43) at (\x-\lx,-\y-\ly) {};

	\draw[draw=alice, thick]  (C11) edge  (C41);
	\draw[draw=bob, thick]  (C22) edge  (C42);
	\draw[draw=charlie, thick]  (C33) edge  (C43);

	\node[rectangle,draw=none,scale=\s] (a12) at (\lx+\la, 0) {$| b \rangle$};
	\node[rectangle,draw=none,scale=\s] (a13) at (\lx+\la, -\ly) {$| c \rangle$};
	
	\node[rectangle,draw=none,scale=\s] (a21) at (\lx+\la, -\y+\ly) {$| a \rangle$};
	\node[rectangle,draw=none,scale=\s] (a23) at (\lx+\la, -\y-\ly) {$| c \rangle$};
	
	\node[rectangle,draw=none,scale=\s] (a31) at (\lx+\la, -2*\y+\ly) {$| a \rangle$};
	\node[rectangle,draw=none,scale=\s] (a32) at (\lx+\la, -2*\y) {$| b \rangle$};

\end{scope}

\begin{scope}[scale=1,shift={(1.5*\x,0)}]

	\node[draw=none](C1) at (0,0) {$\langle \psi |$};
	\node[draw=none] (C2) at (0,-\y) {$\langle \psi |$};
	\node[draw=none] (C3) at (0,-2*\y) {$\langle \psi |$};
	
	\node[draw=none] (C4) at (\x,-\y) {$| \psi \rangle$};
	
	\node[fill=black,circle,inner sep=\r,draw] (C11) at (\lx,0+\ly) {};
	\node[fill=bob,circle,inner sep=\r,draw] (C12) at (\lx,0) {};
	\node[fill=charlie,circle,inner sep=\r,draw] (C13) at (\lx,0-\ly) {};
	
	\node[fill=alice,circle,inner sep=\r,draw] (C21) at (\lx,-\y+\ly) {};
	\node[fill=black,circle,inner sep=\r,draw] (C22) at (\lx,-\y) {};
	\node[fill=charlie,circle,inner sep=\r,draw] (C23) at (\lx,-\y-\ly) {};
	
	\node[fill=alice,circle,inner sep=\r,draw] (C31) at (\lx,-2*\y+\ly) {};
	\node[fill=bob,circle,inner sep=\r,draw] (C32) at (\lx,-2*\y) {};
	\node[fill=black,circle,inner sep=\r,draw] (C33) at (\lx,-2*\y-\ly) {};
	
	\node[fill=black,circle,inner sep=\r,draw] (C41) at (\x-\lx,-\y+\ly) {};
	\node[fill=black,circle,inner sep=\r,draw] (C42) at (\x-\lx,-\y) {};
	\node[fill=black,circle,inner sep=\r,draw] (C43) at (\x-\lx,-\y-\ly) {};

	\draw[draw=alice, thick]  (C11) edge  (C41);
	\draw[draw=bob, thick]  (C22) edge  (C42);
	\draw[draw=charlie, thick]  (C33) edge  (C43);
	
	\path[densely dashed,draw=bob, thick] (C12)  edge [bend left=45] (C32);
	\path[densely dashed,draw=alice, thick] (C21)  edge [bend left=45] (C31);
	\path[densely dashed,draw=charlie, thick] (C13)  edge [bend left=45] (C23);

\end{scope}

\begin{scope}[scale=1,shift={(0,-1.5*\x)}]

	\node[fill=none,circle,inner sep=4*\r,draw](C1) at (0,0) {};
	\node[fill=none,circle,inner sep=4*\r,draw] (C2) at (0,-\y) {};
	\node[fill=none,circle,inner sep=4*\r,draw] (C3) at (0,-2*\y) {};
	
	\node[fill=none,circle,inner sep=4*\r,draw] (C4) at (\x,-\y) {};
	
	\node[draw=none,scale=\s,inner sep=\r] (l12) at (-\w,+\w) {b};
	\node[draw=none,scale=\s,inner sep=\r] (l13) at  (-\w,-\w) {c};
	
	\node[draw=none,scale=\s,inner sep=\r] (l21) at   (-\w,-\y+\w) {a};
	\node[draw=none,scale=\s,inner sep=\r] (l23) at  (-\w,-\y-\w) {c};
	
	\node[draw=none,scale=\s,inner sep=\r] (l31) at  (-\w,+-2*\y+\w) {a};
	\node[draw=none,scale=\s,inner sep=\r] (l32) at  (-\w,-2*\y-\w) {b};
		
	\draw[draw=alice, thick]  (C1) edge node[above,scale=\s] {a} (C4);
	\draw[draw=bob, thick]  (C2) edge node[above,scale=\s] {b} (C4);
	\draw[draw=charlie, thick]  (C3) edge node[above,scale=\s] {c} (C4);
	
	\draw[draw=bob, thick] (C1)  edge (l12);
	\draw[draw=charlie, thick] (C1)  edge (l13);
	
	\draw[draw=alice, thick] (C2)  edge (l21);
	\draw[draw=charlie, thick] (C2)  edge (l23);
	
	\draw[draw=alice, thick] (C3)  edge (l31);
	\draw[draw=bob, thick] (C3)  edge (l32);

\end{scope}

\begin{scope}[scale=0.6,shift={(2.2*\x,-2.2*\x)}]

	\tikzmath{
	\x=0.8; 
	\y=0.6; 
	\yy=0.8; 
	\r=0.35pt; 
	\w=0.2; 
	} 
	
\begin{scope}[scale=1, every node/.style={fill=none,circle,inner sep=4*\r,draw}, every edge/.style={draw=black, thick}]
	
	\node(N1) at (0,0) {};
	\node(N2) at (\x,0) {};
	\node(N3) at (2*\x,0) {};
	\node(N4) at (3*\x,0) {};
	\node(N5) at (4*\x,0) {};
	\node(N6) at (5*\x,0) {};
	
	\draw (N1)  edge (N2);
	\draw (N2)  edge (N3);
	\draw (N3)  edge (N4);
	\draw (N4)  edge (N5);
	\draw (N5)  edge (N6);
	
	\draw (N1) edge  (-\w,+\w) ;
	\draw (N1) edge  (-\w,-\w) ;
	
	\draw (N2) edge  (\x,-\w) ;
	\draw (N3) edge  (2*\x,-\w) ;
	\draw (N4) edge  (3*\x,-\w) ;
	\draw (N5) edge  (4*\x,-\w) ;
	
	\draw (N6) edge  (5*\x+\w,+\w) ;
	\draw (N6) edge  (5*\x+\w,-\w) ;

\end{scope}

\begin{scope}[scale=1,shift={(0.5*\x,-\yy)}, every node/.style={fill=none,circle,inner sep=4*\r,draw}, every edge/.style={draw=black, thick}]
	
	\node(N1) at (0,0) {};
	\node(N2) at (\x,0) {};
	\node(N3) at (2*\x,0) {};
	\node(N4) at (3*\x,0) {};
	\node(N5) at (4*\x,0) {};
	
	\node(N6) at (2*\x,-\y) {};
	
	\draw (N1)  edge (N2);
	\draw (N2)  edge (N3);
	\draw (N3)  edge (N4);
	\draw (N4)  edge (N5);
	
	\draw (N3)  edge (N6);

	\draw (N1) edge  (-\w,+\w) ;
	\draw (N1) edge  (-\w,-\w) ;
	
	\draw (N2) edge  (\x,-\w) ;
	\draw (N4) edge  (3*\x,-\w) ;
	
	\draw (N5) edge  (4*\x+\w,+\w) ;
	\draw (N5) edge  (4*\x+\w,-\w) ;
	
	\draw (N6) edge  (2*\x-\w,-\y-\w) ;
	\draw (N6) edge  (2*\x+\w,-\y-\w) ;

\end{scope}

\begin{scope}[scale=1,shift={(\x,-2.5*\yy)}, every node/.style={fill=none,circle,inner sep=4*\r,draw}, every edge/.style={draw=black, thick}]
	
	\node(N1) at (0,0) {};
	\node(N2) at (\x,0) {};
	\node(N3) at (2*\x,0) {};
	\node(N4) at (3*\x,0) {};
	
	\node(N5) at (3.5*\x,-0.5*\x) {};
	\node(N6) at (3.5*\x,0.5*\x) {};
	
	\draw (N1)  edge (N2);
	\draw (N2)  edge (N3);
	\draw (N3)  edge (N4);
	
	\draw (N4)  edge (N5);
	\draw (N4)  edge (N6);

	\draw (N1) edge  (-\w,+\w) ;
	\draw (N1) edge  (-\w,-\w) ;
	
	\draw (N2) edge  (\x,-\w) ;
	\draw (N3) edge  (2*\x,-\w) ;
	
	\draw (N5) edge  (3.5*\x+\w,-0.5*\x+\w) ;
	\draw (N5) edge  (3.5*\x+\w,-0.5*\x-\w)  ;
	
	\draw (N6) edge  (3.5*\x+\w,0.5*\x+\w);
	\draw (N6) edge  (3.5*\x+\w,0.5*\x-\w) ;

\end{scope}

\begin{scope}[scale=1,shift={(1.5*\x,-4*\yy)}, every node/.style={fill=none,circle,inner sep=4*\r,draw}, every edge/.style={draw=black, thick}]
	
	\node(N1) at (0,-0.5*\x) {};
	\node(N2) at (0,0.5*\x) {};
	
	\node(N3) at (0.5*\x,0) {};
	\node(N4) at (1.5*\x,0) {};
	
	\node(N5) at (2*\x,-0.5*\x) {};
	\node(N6) at (2*\x,0.5*\x) {};

	\draw (N1)  edge (N3);
	\draw (N2)  edge (N3);
	
	\draw (N3)  edge (N4);
	
	\draw (N4)  edge (N5);
	\draw (N4)  edge (N6);

	\draw (N1) edge  (-\w,-0.5*\x+\w) ;
	\draw (N1) edge  (-\w,-0.5*\x-\w) ;
	
	\draw (N2) edge  (-\w,0.5*\x+\w) ;
	\draw (N2) edge  (-\w,0.5*\x-\w) ;

	\draw (N5) edge  (2*\x+\w,-0.5*\x+\w) ;
	\draw (N5) edge  (2*\x+\w,-0.5*\x-\w)  ;
	
	\draw (N6) edge  (2*\x+\w,0.5*\x+\w);
	\draw (N6) edge  (2*\x+\w,0.5*\x-\w) ;

\end{scope}

\end{scope}

\end{tikzpicture}
    \end{minipage}
    \caption{Graphical description of the possible estimates in the second 
    hierarchy $\mathfrak{H}_2$.
    Upper left: a visualization of Eq.~\eqref{eq:2nd_hierarchy_example}.
    Upper right: relaxation from product states; the dashed lines denote the application of a projector onto the symmetric subspace.
    Lower left: resulting connectivity graph; the short legs denote the indices on which the symmetric projector acts.
    Lower right: possible connectivity graphs for six copies of a tripartite state. Note that different labelings are possible which may lead to different results.
    }
    \label{fig:second_hierarchy_wirings}
\end{figure}


In the third hierarchy $\mathfrak{H}_3$, instead of taking multiple copies of
$\ket{\psi}$, we take just one copy of $\ket{\psi}$ and tensor product it
with $k-1$ identity operators, i.e.,
\begin{equation}
   \Lambda^2(\psi) 
    =\max_{\ket{abc}}  [\bra{abc}^{\otimes k}]
    [\ketbra{\psi}{\psi}\otimes\openone^{\otimes k-1}]
    [\ket{abc}^{\otimes k}].
\end{equation}
Again, by taking advantage of the fact that the projector on the symmetric subspace can be inserted without changing the above expression, we get an upper bound on $\Lambda^2$ by the maximal eigenvalue of the above operator, i.e.,
\begin{equation}
   \Lambda^2(\psi) 
    \le\lambda_{\max} \qty(\Pi_k^{\otimes 3}
    [ \ketbra{\psi}{\psi}\otimes\openone^{\otimes k-1}]
    \Pi_k^{\otimes 3}).
\end{equation}
Using the quantum de Finetti theorem, we can prove that
the hierarchy $\mathfrak{H}_3$ is also complete. For the 
proof and implementation details we refer to Appendices~\ref{app:proof_convergence_third_hierarchy} and \ref{app:third_numerical}.

\section{Results}\label{sec:results}
As a first example of the three hierarchies, we bound the geometric measure for superpositions of the three-qubit 
W and GHZ state,
$\ket{\psi(s)}=\sqrt{s}\ket{W} + \sqrt{1-s}\ket{\mathrm{GHZ}}$,
with 
$\ket{W} = (\ket{001}+ \ket{010}+\ket{100})/\sqrt{3}$ 
and 
$\ket{\mathrm{GHZ}}=
 (\ket{000}+ \ket{111})/\sqrt{2}$, see also Fig.~\ref{fig:W_vs_GHZ}. The levels used for the three hierarchies were $25$, $10$, and $60$,  respectively. All three hierarchies are numerically tight for 
 a wide range of $s$, while the gap increases for the more entangled states. In particular $\mathfrak{H}_3$ is nearly indistinguishable from the correct value
 over the whole parameter range.

\begin{figure}[t!]
    \centering
    \includegraphics[trim={0mm 2.1mm 0mm 0mm}, clip, width=0.9\linewidth]{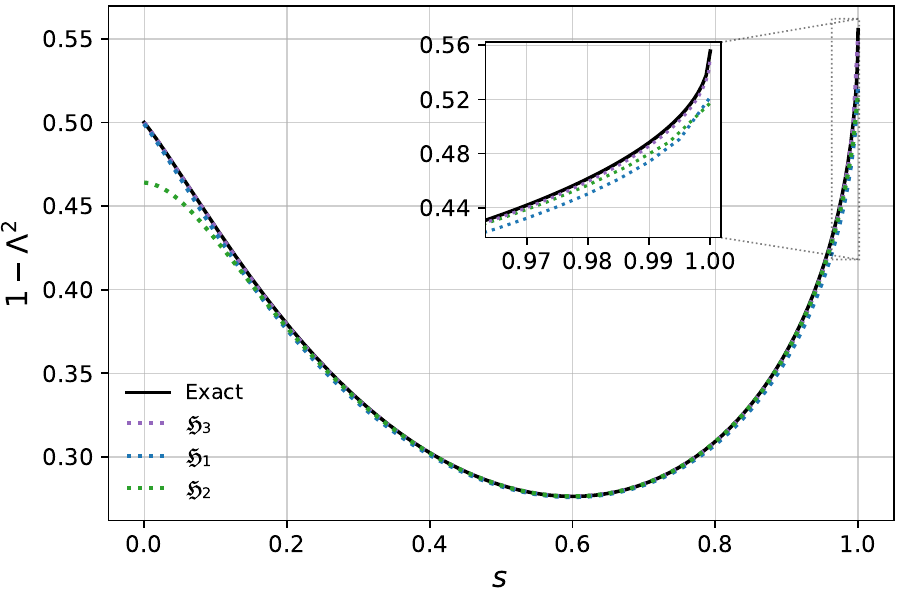}
\caption{Lower bounds from the three hierarchies for states of the form $\ket{\psi(s)}=\sqrt{s}\ket{W} + \sqrt{1-s}\ket{\mathrm{GHZ}}$. We also show 
the exact value obtained by solving the one-parameter optimization over
symmetric product states \cite{wei2003geometric}. For the first hierarchy $\mathfrak{H}_1$ we
use level $25$ and consider the Schmidt decomposition of 
$\ket{F_k}$ with respect to the bipartition between the 
first copy and the rest of the copies for improvement. 
For $\mathfrak{H}_2$, for which we only consider the norm
and level $10$, the used connectivity graph is described
in Appendix~\ref{app:trees}. For $\mathfrak{H}_3$ a level of $60$
was used.
}
\label{fig:W_vs_GHZ}
\end{figure}

The hierarchies can also handle larger number of particles. 
To this end, we show the rates of convergence for the different hierarchies for $\ket{C_5}$, the $5$-cycle graph state, in Fig.~\ref{fig:convergence_C5}. This state is the maximally entangled state with respect to the geometric measure on 
five qubits having a value of $E_G=0.86855$\,\cite{steinberg2024finding}.

\begin{figure}[t]
    \centering
    \includegraphics[trim={0mm 2.1mm 0mm 0mm}, clip, width=0.9\linewidth]{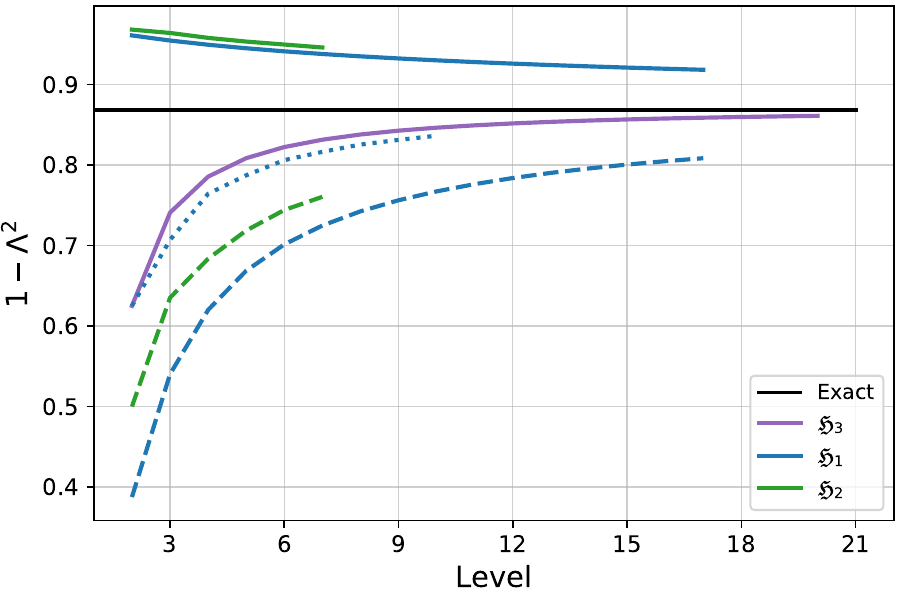}
\caption{Bounds on the geometric measure of the $5$-cycle graph state $\ket{C_5}$, which is the maximally entangled state w.r.t.~the geometric measure for five qubits. The horizontal line at height $0.86855$ 
is the value of the geometric measure of $\ket{C_5}$~\cite{steinberg2024finding}. For $\mathfrak{H}_1$ and $\mathfrak{H}_2$ we show the upper and lower bounds in solid and dashed, respectively. For $\mathfrak{H}_1$ we also show the lower bound found by considering the largest Schmidt coefficient of $\ket{F_k}$ with respect to a balanced bipartition between the copies, i.e., the bipartition between the first and second half of the copies. Note that the numerical package ENTCALC \cite{masajada2025entcalctoolkitcalculatinggeometric}
provides for this state only a certified bound in the interval $[0.7292, 0.9375]$. 
}
\label{fig:convergence_C5}
\end{figure}

\section{Relation to the product test}
The construction of the hierarchies above is in close connection 
to the so-called product test for entanglement \cite{mintert2005concurrence,beausoleil2006test,harrow2013testing,foulds2021controlled}. For this, one 
first considers the so-called SWAP test for two particles 
$\varrho_{A_1}$ and $\varrho_{A_2}$, which asks whether the
two states are the same and is passed with success probability
$1/2 + \Tr(\varrho_{A_1}\varrho_{A_2})/2 = 
\Tr(\Pi_2\varrho_{A_1}\otimes \varrho_{A_2}).$ Then, if 
two copies of a pure multiparticle state $\ket{\psi}$ are 
given, one can apply this SWAP test to each pair of particles
$A_1 A_2$, $B_1 B_2$ etc.~locally and it is clear that the 
two states pass this test with unit probability if and only 
if $\ket{\psi} = 
\ket{a}\ket{b}\ket{c}$ is a product state. Several works 
investigated the statistical significance of the test
for special cases: For two copies of multiparticle states 
it was shown that the quantity
$\Lambda^2$ is a lower bound on the passing probability
\cite{harrow2013testing}, and for two particles and more copies 
it was demonstrated that the passing probability converges
to $\Lambda^2$ if the number of copies increases \cite{she2022unitary}.

The completeness of the hierarchy $\mathfrak{H}_1$ naturally explains and generalizes these results. In fact, $\Vert\ket{F_k}\Vert^2$ can be understood as the probability 
that the multipartite state $\ket{\psi}$ passes the 
multi-copy product test and the completeness in Eq.~(\ref{eq-observation1}) gives upper and lower
bounds on this probability in terms of the 
entanglement of the state
(see also Eq.~\eqref{eq:boundProducttest} in Appendix~\ref{app:proof_convergence_first_hierarchy}).

Interestingly, the fact that the hierarchy $\mathfrak{H}_2$ gives strictly better upper bounds on $\Lambda^2$ than $\mathfrak{H}_1$ (when considering solely the norm in both cases) shows that, if several copies are available, strictly better tests than the standard multi-copy product test are available.

We note that there are several proposals to implement the product test using 
quantum circuits \cite{bradshaw2023cycle,liu2025MeasEntSym,liu2025generalized} or randomized measurements \cite{Elben2020MixedRandMeas,Cieslinski2024RandMeasReview}, 
these methods are also appropriate for implementing the generalized 
product tests based on our hierarchies. Finally, we note that the 
two-copy scenario described above is reminiscent of Bell sampling, 
where two copies of a state $\rho$ are pairwise measured in the Bell
basis~\cite{montanaro2017learning}.

\section{Application to Operators}
Interestingly, the  above hierarchies can be adopted also to 
operators, considering the related notion of the maximal separable 
numerical range $M(X) = \max_{\ket{a,b,c}} \bra{abc}X\ket{abc}$
where $X$ is a general  operator \cite{puchala2011product}. 
This quantity arises e.g.~in the construction of entanglement 
witnesses and the distillation problem, as well as in problems in many-body physics \cite{dowling2004energy}. Generalizing the first 
hierarchy $\mathfrak{H}_1$, one then finds $M^k(X) \leq \lambda_{\text{max}} \qty( \Pi_k^{\otimes 3} X^{\otimes k}\Pi_k^{\otimes 3} )$,
which also converges for positive operators to the correct value in 
the limit $k\rightarrow \infty$, as can be seen by considering the 
purification $\ket{\psi}$ of $X$ on four parties, where $\Lambda^2(\psi)=M(X)$
\cite{harrow2013testing, weinbrenner2025quantifying}.
An analogous result holds for the third hierarchy, $\mathfrak{H}_3$, where
convergence is guaranteed for any $X$.

As a simple, yet nontrivial example for two parties we consider the 
projector $P_{\text{UPB}}$ on an unextendible product basis (UPB) 
given in~\cite{bennett1999unextendible} (and recounted in 
Appendix~\ref{app:UPB}). Since there is no product 
vector in the space orthogonal to an UPB, the state $X_{\text{UPB}} \propto\id-P_{\text{UPB}}$ is an entangled state. For the witness construction one needs then the quantity $\delta=M(X_{\text{UPB}})$. Numerically this value was shown 
to be $\delta\approx 0.97158$ \cite{guhne2003experimental}, analytically 
only the bound $\delta\leq 0.998703$ is known \cite{terhal2001family}. 
Using the first hierarchy we find for $k=7$ a bound of $\delta \leq 0.976057$ 
and with the third hierarchy for $k=11$ a bound of $\delta \leq 0.973382$, 
improving in both cases the previous analytical bound significantly. 

More generally, the extension to operators sheds some light on the complexity
of the separability problem itself. Generally, deciding whether a mixed state 
admits a separable decomposition is Nondeterministic Polynomial-time (NP) hard~\cite{gurvits2003entNPhard,gharibian2010strong}. In practice, it can
be tackled by semidefinite programming (SDP) hierarchies approximating the 
separable set~\cite{Doherty2004complete}, where the dimension of the 
SDP increases in each step of the hierarchy. Consequently, it is clear
that $M(X)$ can be computed by considering a similar hierarchy. Our results show, 
however, that the computation of $M(X)$ (or, equivalently, checking whether
a given observable is an entanglement witness) can also be determined
by a hierarchy of eigenvalue problems, which may reduce the complexity.

A further possible application considers the optimization over operators. For example, one might consider an SLOCC-class $S_{\ket{\phi}}$ (stochastic local operations and classical communications), consisting of all states that can be reached from $\ket{\phi}$ by SLOCC-operations, i.e., $A_1\otimes \dots\otimes A_n\ket{\phi}$. For a given state $\ket{\psi}$ one might be interested in the maximal achievable overlap with states from $S_{\ket{\phi}}$, leading to the value $\max_{A_1\otimes\dots A_n} |\bra{\psi}A_1\otimes\dots \otimes A_n \ket{\phi}|$. Interestingly, this problem can be rephrased as an optimization over product states, see, e.g., Refs.~\cite{kampermann2012algorithm,ritz2019characterizing}.

\section{Generalization to mixed states}
So far, we considered the geometric measure for pure states only. In the 
real world, however, noise is 
unavoidable, leading to the occurrence of mixed states 
or density matrices. For such a density matrix $\varrho$
the geometric measure is defined via the convex roof 
construction,
\begin{equation}
E_G(\varrho) = \min_{p_k, {\psi_k}} \sum_k p_k E_G (\psi_k)    
\end{equation}
where the minimum is taken over all decompositions 
$\varrho= \sum_k p_k \ketbra{\psi_k}{\psi_k}$ of the mixed states
into pure states $\psi_k$  with some probabilities $p_k$. Naturally,
this minimization is difficult to compute. Upper bounds on 
$E_G(\varrho)$ may be obtained by brute-force numerical 
optimization, but specifically lower bounds are difficult and
only few exact results are known \cite{guhne2008multiparticle,buchholz2016evaluating,masajada2025entcalctoolkitcalculatinggeometric}.

Taking the viewpoint of a physicist,  it is a critical question whether 
the presented hierarchy can give insights into the quantification
of entanglement for mixed states. In the following, using the approach
of Ref.~\cite{toth2015evaluating} we will demonstrate that this is indeed the case.  
To see this, consider the first step of the hierarchy for three 
particles, as presented in Eq.~\eqref{eq:genidea2}. This implies that the 
geometric measure for pure states $\varrho=\ketbra{\psi}{\psi}$ is bounded by 
\begin{equation}
E(\psi) \geq \tr(G_2 \varrho^{\otimes 2})    
\end{equation}
with $G_2 = \openone - \Pi_A \otimes \Pi_B \otimes \Pi_C$. If $\varrho$
is a mixed state, then any decomposition 
$\varrho= \sum_k p_k \ketbra{\psi_k}{\psi_k}$ is in one-to-one 
correspondence to a two-copy state 
$\sigma_{1,2} = \sum_k p_k \ketbra{\psi_k}{\psi_k} \otimes \ketbra{\psi_k}{\psi_k}$
which is separable (with respect to the partition coming from the
two copies), in the symmetric subspace and reduces to $\varrho$ via $\tr_2(\sigma_{1,2}) = \varrho.$
So it follows that
\begin{align}
\label{eq-mixed-bound}
E(\varrho)  \geq & \min_{\sigma_{1,2}}\tr(G_2 \sigma_{1,2})   
 \\
& \mbox{s.t.:  } \sigma_{1,2} \mbox{ separable, symmetric and }
\tr_2(\sigma_{1,2}) = \varrho.
\nonumber
\end{align}
The point is that the separability condition can be relaxed by requiring
that $\sigma_{1,2}$ is only PPT with respect to the $1|2$ partition and
the resulting lower bound can be directly computed via semidefinite programming (SDP). This approach can be generalized to higher orders
of our hierarchy or to other relaxations of the separability constraint
in Eq.~(\ref{eq-mixed-bound}). Note finally that there are other approaches
that can be used to convert our results for pure states to characterizing
mixed states \cite{mintert2007entanglement}.

In fact, already the simple PPT relaxation of the lowest order in 
Eq.~(\ref{eq-mixed-bound}) gives surprisingly strong criteria for 
entanglement. As a first example, we considered GHZ states mixed 
with white noise $\varrho(p) = p \ketbra{\mathrm{GHZ}}{\mathrm{GHZ}} + (1-p) \openone/8$. 
These are known to be entangled if and only if $p > 1/5$~\cite{schack2000explProdEnsSEP,Dur2000ClassMultiQubit} and the SDP directly reproduced this threshold. The second example
are W states mixed with noise, $\varrho(p) = p \ketbra{W}{W} + (1-p) \openone/8$. These states are known to be fully separable for 
$p \leq 0.177$, entangled for $p>\sqrt{3}/(8+\sqrt{3}) \approx 0.1778$,
but still separable for any bipartition until $p \leq 0.2095$~\cite{chen2012estimating}. 
This means that entanglement in the regime $p \in [0.1778, 0.2095]$ cannot be detected 
by any bipartite criterion, such as the PPT criterion. Our SDP, however, directly allows to detect 
the entanglement already with the second level of the hierarchy for all
$p \geq 0.1781$, demonstrating the ability to detect also weakly entangled
states in a straightforward manner. Finally, the three-qubit Tao state $\tau$---which is the most robust entangled state that is biseparable for any bipartition~\cite{ohst2024certifying}---can be detected in a parameter regime where all other known criteria fail, further details can be found in Appendix~\ref{app:Mixed}. Finally, we add that the presented approach can also be used to
characterize the convex roof of other entanglement monotones 
\cite{beckey2021computable, schatzki2024hierarchy} which
attracted some attention recently and have a formal similarity to the steps of the hierarchies.

\section{Discussion}
We introduced three different hierarchies for the geometric measure of entanglement and showed their convergence, considering multiple copies of the target state and applying different symmetrization arguments. On the mathematical side, this implies a characterization of the injective tensor norm as a limit of the 2-norms of a family of vectors. On the physical side, this leads to easily computable upper and lower bounds on the geometric measure of entanglement. We find that the first hierarchy $\mathfrak{H}_1$ connects closely with the product test, extending results known for the bipartite case to the multipartite setting. The third hierarchy we find to provide the best lower bounds in practice.

An interesting open question, which initiated also the
mathematical discussion \cite{helfgott}, is whether the hierarchies can also be interpreted in a combinatorial way. The known upper bound for symmetric matrices can be understood as counting the closed walks in a weighted graph, where the weights are described by the matrix. For the case of complex tensors, such an interpretation 
requires still some work, but may result in an abundant source of further inspiration. 

A more physical application is the maximization of state overlaps
over local unitaries, which is 
also an optimization over product structures, see e.g. \cite{castellano2025parallel}. The hierarchies seem to be relevant also in this case, although a more detailed analysis is needed. For example, 
in the case of unitary qubit operators one optimizes implicitly over real 
instead of complex states \cite{Neven2016FidSymStates}, so it might be worth investigating further the relationship between the optimization over real and complex states.
Mathematically, the geometric measure can also be seen as the minimal distance of a higher rank tensor to the rank-1 tensors. Naturally, this can be generalized to higher ranks, for example, asking for the best rank-2 approximation of a given tensor. It would be interesting to investigate whether similar hierarchies can be formulated for these cases.

{\it Acknowledgments---}%
We thank Riccardo Castellano, Sophia Denker, Shmuel Friedland, Sev Gharibian, 
Ties-Albrecht Ohst, Jordi Romero-Pallejà, Dorian Rudolph, Rene Schwonnek for discussions, and especially Nikolai Wyderka for help with the proof of the 
real injective tensor norm in Appendix E.
This work was supported by the Deutsche Forschungsgemeinschaft 
(DFG, German Research Foundation, project number 563437167), the 
Sino-German Center for Research Promotion (Project M-0294),
and 
the German Federal Ministry of Research, Technology and Space (Project QuKuK, Grant No. 16KIS1618K and Project
BeRyQC, Grant No. 13N17292).
LTW acknowledges the support from the House of Young Talents of the University of Siegen. 
AR acknowledges financial support from Spanish MICIN (projects: PID2022; 141283NBI00; 139099NBI00) with the support of FEDER funds, the Spanish Government with funding
from European Union NextGenerationEU (PRTR-C17.I1), the Generalitat de Catalunya,
the Ministry for Digital Transformation and of Civil Service of the Spanish Government through the QUANTUM ENIA project -Quantum Spain Project- through the Recovery, Transformation and Resilience Plan Next-Generation EU within the framework
of the Digital Spain 2026 Agenda. 
KG acknowledges the support from the Alexander von Humboldt Foundation. 
XDY acknowledges the support from the National Natural Science Foundation of China
(Grants No.~12574537 and No.~12205170)
and the Shandong Provincial Natural Science Foundation of China (Grant No. ZR2022QA084).

\onecolumngrid
\section*{APPENDIX}
\appendix


We prove in Appendix \ref{app:proof_convergence_first_hierarchy} the completeness of the first hierarchy $\mathfrak{H}_1$; in Appendix \ref{subsec:Herarchy1} we show how to reduce the computation time and include the restriction to PPT states for $\mathfrak{H}_1$. Appendix \ref{app:trees} contains the implementation details for the second hierarchy $\mathfrak{H}_2$, and we prove its completeness in Appendix \ref{app:proof_convergence_second_hierarchy}. We detail the connection between our results and the ideas from H.~A.~Helfgott and S.~Friedland in Appendix \ref{app:connection_helfgott_friedland}. Completeness and implementation details of the third hierarchy $\mathfrak{H}_3$ are shown in Appendices \ref{app:proof_convergence_third_hierarchy} and \ref{app:third_numerical}. Appendix \ref{app:UPB} discusses how hierarchies $\mathfrak{H}_1$ and $\mathfrak{H}_3$ can be used to find the separable numerical range of operators. We conclude by finding lower bounds on the convex roof extension of $E_G$ for mixed states in Appendix \ref{app:Mixed}.

\section{The first hierarchy $\mathfrak{H}_1$: proof of Observation~\ref{obs:convergence}}\label{app:proof_convergence_first_hierarchy}

\noindent
{\bf Observation~\ref{obs:convergence}.} {\em Given an $n$-partite state $\ket{\psi}$, and for any $k$, we have the two-sided bound
\begin{equation}
    d_k \Vert \ket{F_k} \Vert^{2/k}
    \leq 
    \Lambda^2(\psi)
    \leq
    \Vert \ket{F_k} \Vert^{2/k},
\end{equation}
where the coefficients $d_k:=\binom{k+d-1}{k}^{-n/k}<1$ 
converge to one, $\lim_{k\rightarrow\infty} d_k =1$. 
}

\begin{proof}
For a general number of copies $k$, we define
$ \ket{F_k} = \Pi_k\otimes \Pi_k\otimes \Pi_k [\ket{\psi}^{\otimes k}]$, where the $k$-copy symmetric projector is given by $\Pi_k=\sum_{\pi\in S_k}V_\pi/k!$ with $S_k$ being the symmetric group of $k$-element permutations $\pi$ and $V_\pi$ their representation on $(\mathbb{C}^d)^{\otimes k}$. This definition for three-partite states extends naturally to $n$-partite states, through $ \ket{F_k} = \Pi_k^{\otimes n} [\ket{\psi}^{\otimes k}]$. 

As sketched in the main text, we first note that 
\begin{align}
    \Lambda(\psi) :&=\max_{\phi\in\text{SEP}}|\langle\phi\ket{\psi}|
\nonumber
\\ 
    &= \max_{\phi\in\text{SEP}}|\langle\phi|^{\otimes k}\ket{\psi}^{\otimes k}|^{1/k}
    \nonumber
    \\
    &= \max_{\phi\in\text{SEP}}|\langle\phi|^{\otimes k}\Pi_k^{\otimes n}\ket{\psi}^{\otimes k}|^{1/k} \nonumber \\
    &\leq \max_{\Phi\in{(\mathbb{C}^d)}^{n k}}|\bra{\Phi}\Pi_k^{\otimes n}\ket{\psi}^{\otimes k}|^{1/k}\nonumber \\
    &= ||\ket{F_k}||^{1/k}\,.
\end{align}
The third equality holds because $\Pi^{\otimes n}_k\ket{\phi}^{\otimes k}=\ket{\phi}^{\otimes k}$ for any separable state $\ket{\phi}$. The inequality holds because the maximization over general states $\ket{\Phi}$ contains in particular $\ket{\phi}\in\text{SEP}$, and the next equality holds because the state $\ket{\Phi}$ with maximal overlap with $\ket{F_k}$ is proportional to it, $\ket{\Phi}\propto\ket{F_k}$. This proves the right-hand-side inequality in Observation~\ref{obs:convergence}. 

\bigskip

To prove the left-hand-side inequality in Observation~\ref{obs:convergence}, we will write the symmetric projector as
\begin{equation}
    \Pi_k=\binom{k+d-1}{k}\int_{\phi\in \mathrm{Haar}}\ket{\phi}\bra{\phi}^{\otimes k}d\phi\,,
\end{equation}
where Haar is the flat Haar-random measure on pure states in ${\mathbb{C}^d}$ \cite{harrow2013church}.
Therefore, $n$ copies of $\Pi_k$ read
\begin{equation}
    \Pi_k^{\otimes n}=\binom{k+d-1}{k}^n\int_{\phi_1,...,\phi_n\in \mathrm{Haar}} \ket{\phi_1...\phi_n}\bra{\phi_1...\phi_n}^{\otimes k}d\phi_1...d\phi_n\,.
\end{equation}
Denote $\ket{a_1...a_n}$ the closest separable state to $\ket{\psi}$, namely $\Lambda(\psi)=|\bra{a_1...a_n}\psi\rangle|$. Then the norm $||\Pi_k^{\otimes n}\ket{\psi}^{\otimes k}||^2$ reads
\begin{align}\label{eq:first_hierarchy_conv_proof}
    ||\Pi_k^{\otimes n}\ket{\psi}^{\otimes k}||^2 &= \bra{\psi}^{\otimes k}\Pi_k^{\otimes n} \ket{\psi}^{\otimes k}
    \nonumber\\
    &= \binom{k+d-1}{k}^n\int_{\phi_1,...,\phi_n\in \mathrm{Haar}} |\bra{\phi_1...\phi_n}\psi\rangle|^{2 k}d\phi_1...d\phi_n\nonumber\\
    &\leq \binom{k+d-1}{k}^n|\bra{a_1...a_n}\psi\rangle|^{2 k}\int_{\phi_1,...,\phi_n\in \mathrm{Haar}} d\phi_1...d\phi_n\nonumber\\
    &= \binom{k+d-1}{k}^n|\bra{a_1...a_n}\psi\rangle|^{2 k}\nonumber\\
    &= \binom{k+d-1}{k}^n\Lambda^{2 k}(\psi)
\end{align}
since $\int_{\phi_1,...,\phi_n\in \mathrm{Haar}} d\phi_1...d\phi_n=1$. Taking the $k$'th root at both sides we have that
\begin{equation}
    ||\Pi_k^{\otimes n}\ket{\psi}^{\otimes k}||^{2/k} \leq \binom{k+d-1}{k}^{n/k}\Lambda^{2}(\psi)\,.
\end{equation}
Defining $d_k:=\binom{k+d-1}{k}^{-n/k}$, and observing that $d_k<1$ for finite $k$ by definition, this means that
\begin{equation}
    d_k||\ket{F_k}||^{2/k}\leq\Lambda^2(\psi)\,,
\end{equation}
which proves the left-hand-side inequality in Observation~\ref{obs:convergence}.

\bigskip

It is only left to see convergence of the hierarchy, namely that $d_k$ tends to $1$ in the limit $k\rightarrow\infty$. 
For this we note that the expression $\binom{k+d-1}{k}=\frac{1}{(d-1)!}\prod_{i=1}^{d-1}(k+i)$ is a polynomial in $k$ of degree $d-1$, so the leading terms of $d_k$ are of the order $k^{-\frac{dn}{k}}$. We therefore find
\begin{align}
    d_k = \binom{k+d-1}{k}^{-n/k} = O\left(k^{-\frac{dn}{k}}\right)\rightarrow 1 \text{ for } k\rightarrow \infty. \label{eq:dk_convergence}
\end{align}

Therefore the hierarchy is convergent for finite dimension $d$ and number of parties $n$.
\end{proof}

Note that the above observation can also be phrased as bounding the norm $\norm{\ket{F_k}}$, which is related to the acceptance probability in the context of the product test. In this formulation the result reads
\begin{align}\label{eq:boundProducttest}
    \Lambda^{2k}(\psi) \leq \norm{\ket{F_k}}^2
    \leq \frac{1}{(d_k)^{2k}} \Lambda^{2k}(\psi).
\end{align}

Note that the first hierarchy can also be formulated for operators, see Appendix~\ref{app:UPB}. In the case of positive operators one finds that the hierarchy also converges, where the proof works in the same way as above.

\section{The first hierarchy $\mathfrak{H}_1$: Numerical approach}\label{subsec:Herarchy1}

We demonstrate here how the bounds from first hierarchy can be calculated in practice, using the three-qubit case as an example.

\subsection{Calculating the norms of $\ket{F_k}$}

In the first hierarchy, we consider the vector 
$\ket{F_k} = \Pi_k\otimes \Pi_k\otimes \Pi_k [\ket{\psi}^{\otimes k}]$. 
As this vector is symmetric on every subsystem A, B, and C, it can be written as 
$\ket{F_k} = \sum_{ijl} c_{ijl} \ket{D_k^i}\ket{D_k^j}\ket{D_k^l}$, 
where 
$\ket{D_k^i} = \frac{1}{N_i} \sum_{\pi \text{ perm.}} \pi(\ket{0}^{\otimes i}\otimes \ket{1}^{\otimes (k-i)})$ are the Dicke states. The sum runs over all distinct permutations of $\ket{0}^{\otimes i}\otimes \ket{1}^{\otimes (k-i)}$ and the normalization constant is then given by $N_i^2=\binom{k}{i}$. As these states form a basis, we need effectively only to compute the overlaps 
$c_{ijk} = |\bra{\psi}^{\otimes k}\ket{D_k^i}\ket{D_k^j}\ket{D_k^l} |$. 

We start by considering a single qubit state $\ket{\phi}$ and the coefficient $\braket{\phi^{\otimes k}}{D_k^i}$. Noticing that the Schmidt decomposition of a Dicke state for the bipartite cut between one party and the rest is given by
$\ket{D_k^i} = \sqrt{\frac{i}{k}} \ket{0}\otimes \ket{D_{k-1}^{i-1}} + \sqrt{\frac{k-i}{k}} \ket{1}\otimes \ket{D_{k-1}^{i}}$
we directly find the iterative rule
\begin{align}
    \braket{\phi^{\otimes k}}{D_k^i} = 
    \sqrt{\frac{i}{k}} \braket{\psi}{0} \braket{\phi^{\otimes k-1}}{D_{k-1}^{i-1}}
    + \sqrt{\frac{k-i}{k}} \braket{\psi}{1} \braket{\phi^{\otimes k-1}}{D_{k-1}^{i}}.
\end{align}
Using this rule one may calculate the coefficients of $\ket{\phi}^{\otimes k}$ with respect to the Dicke basis iteratively while only considering vectors of length up to $k+1$, reducing the computational effort drastically. 
This approach generalizes straightforwardly to the multipartite case of $n$ parties, where one then has to consider only vectors of length up to $(k+1)^n$.
Finally, the norm $\norm{\ket{F_k}}$ can then be directly calculated by $\sum_{i,j,l} |c_{ijl}|^2$.

\subsection{Including PPT in the hierarchies}

As noted in the main text, the three different hierarchies give upper bounds on the geometric measure by relaxing the optimization over product states to arbitrary states. One may, however, retain some of the structure of the product states in the optimization by considering states which have a positive partial transpose (PPT) with respect to a (or every) bipartition, which leads to improved bounds. As an example, we demonstrate here how this approach can be implemented in the first hierarchy $\mathfrak{H}_1$.

In the first hierarchy, we have for $k=2$
\begin{align}
\Lambda^4(\psi)  
= \max_{\ket{abc}} |\bra{abc}^{\otimes 2}\ket{F_2}|^2
=
\max_{\sigma \in \text{SEP}} \Tr[\ketbra{F_2}{F_2} \sigma]
\leq 
\max_{\sigma \in \text{SEP}(S|\bar{S})} \Tr[\ketbra{F_2}{F_2} \sigma]
\leq \Vert \ket{F_2} \Vert^2
\end{align}
with $\ket{F_2} = \Pi_2\otimes \Pi_2\otimes \Pi_2 [\ket{\psi}^{\otimes 2}]$, where $\text{SEP}$ and $\text{SEP}(S|\bar{S})$ denote all fully separable states and all states which are separable with respect to a fixed bipartition $S|\bar{S}$ of the two-copy space $A_1B_1C_1A_2B_2C_2$. For a bipartite state the maximal achievable overlap with a separable state is given by the largest Schmidt coefficient. Note, however, that the vector $\ket{F_2}$ is not normalized, so instead we find as an upper bound $\Tr[\ketbra{F_2}{F_2} \sigma]\leq s_{\max}^{S|\bar{S}}$ for $\sigma$ separable across the bipartition, where $s_{\max}^{S|\bar{S}}$ denotes the largest singular value of the reduced density matrix $\Tr_{\bar{S}}[\ketbra{F_2}{F_2}]$.

These arguments generalize to the case of multiple parties and copies, leading to various possible bipartitions one could consider for an upper bound on $\Lambda(\psi)$.
In the examples considered in this paper, we implemented two different fixed bipartitions. The first is given by considering the cut between the first copy and all other copies, while the second describes the cut between the first half and the second half of the $k$ copies. As the calculation of the largest singular value depends on the size of the reduced density matrix across the cut, the second bipartition is numerically more demanding than the first. However, the second bipartition leads typically to better upper bounds than every other bipartition for randomly chosen states.

We show here shortly how the Schmidt coefficients can be calculated in practice for three-qubit states for the bipartition between the first copy and the remaining copies. Note that the Friedland state $\ket{F_k}$ can be written as $\ket{F_k} = \sum_{ijl} c_{ijl} \ket{D_k^i D_k^j D_k^l}$. The Schmidt coefficient is then given by the maximal eigenvalue of $\Tr_{A_2\dots A_k B_2\dots B_k C_2\dots C_k}[\ketbra{F_k}{F_k}]$. For the Dicke states the reduced density matrix can be calculated directly as 
\begin{align}
    \Delta(k;\alpha,\beta) := \Tr_{2,\dots,k}\left(\ketbra{D_k^\alpha}{D_k^\beta}\right) = \frac{1}{n}\begin{pmatrix}
        \alpha \delta_{\alpha,\beta} & \sqrt{\alpha(k-\beta)} \delta_{\alpha,\beta+1} \\
        \sqrt{\beta (k-\alpha)} \delta_{\alpha,\beta-1} & (k-\alpha) \delta_{\alpha,\beta}
    \end{pmatrix},
\end{align}
leading to 
\begin{align}
    \Tr_{A_2\dots A_k B_2\dots B_k C_2\dots C_k}[\ketbra{F_k}{F_k}] = \sum_{ijl,\alpha\beta\gamma} c^*_{ijl} c_{\alpha\beta\gamma} \Delta(k;i,\alpha) \otimes \Delta(k;j,\beta) \otimes \Delta(k;l,\gamma).
\end{align}
This also is generalizable to more parties, thus, the calculation of the Schmidt coefficients from given coefficients $c_{ij\dots}$ only amounts to the eigenvalue decomposition of a $2^n\times 2^n$ matrix. The computational important step is the construction of this matrix from the coefficients $c_{ij\dots}$, as the number of these grows as $k^n$. In practice, this scheme was for $k\leq 25$ in comparison not much slower than calculating the norm, while leading to a better upper bound.

\section{The second hierarchy $\mathfrak{H}_2$: Numerical approach}
\label{app:trees}
The trees that arise in the second hierarchy can be interpreted as tree tensor networks (TTNs) \cite{montangero2018introduction}, whose individual tensors are copies 
of the state $\ket{\psi}$ under consideration. Two examples for four particles and four copies are given in Fig.~\ref{fig:example_trees}.

\begin{figure}[t]
    \centering
        \begin{minipage}{0.25\textwidth}
        \centering
        \begin{adjustbox}{clip,trim=0cm 0.0cm 0cm 0.0cm,width=\linewidth}
            \begin{tikzpicture}[
  scale=0.5,
  every node/.style={transform shape},
  >=latex,
  tensor/.style={
    circle,
    draw=black!70,
    line width=0.9pt,
    minimum size=9mm,
    shade,
    top color=blue!15,
    bottom color=blue!55,
    preaction={fill=black,opacity=0.15,
      transform canvas={shift={(0.9pt,-0.9pt)}}}
  },
  leg/.style={line width=0.9pt}
]

\coordinate (C)  at ( 0.0,  0.0);   
\coordinate (T1) at ( 1.3,  1.3);   
\coordinate (T2) at (-1.3,  1.3);   
\coordinate (T3) at (-1.3, -1.3);   
\coordinate (T4) at ( 1.3, -1.3);   

\draw[leg] (C) -- node[pos=0.5, yshift=8pt] {$2$} (T1);
\draw[leg] (C) -- node[pos=0.5, yshift=8pt] {$1$} (T2);
\draw[leg] (C) -- node[pos=0.5, yshift=8pt] {$3$} (T3);
\draw[leg] (C) -- node[pos=0.5, yshift=8pt] {$4$} (T4);

\draw[leg] (T1) -- node[pos=1, xshift=-4pt, yshift=4pt,  font=\scriptsize] {$3$} ++(-0.9, 0.9); 
\draw[leg] (T1) -- node[pos=1, xshift= 4pt, yshift=4pt,  font=\scriptsize] {$4$} ++( 0.9, 0.9); 
\draw[leg] (T1) -- node[pos=1, xshift= 4pt, yshift=-4pt, font=\scriptsize] {$1$} ++( 0.9,-0.9); 

\draw[leg] (T2) -- node[pos=1, xshift=-4pt, yshift=4pt,  font=\scriptsize] {$3$} ++(-0.9, 0.9); 
\draw[leg] (T2) -- node[pos=1, xshift= 4pt, yshift=4pt,  font=\scriptsize] {$4$} ++( 0.9, 0.9); 
\draw[leg] (T2) -- node[pos=1, xshift=-4pt, yshift=-4pt, font=\scriptsize] {$2$} ++(-0.9,-0.9); 

\draw[leg] (T3) -- node[pos=1, xshift=-4pt, yshift=4pt,  font=\scriptsize] {$2$} ++(-0.9, 0.9); 
\draw[leg] (T3) -- node[pos=1, xshift=-4pt, yshift=-4pt, font=\scriptsize] {$3$} ++(-0.9,-0.9); 
\draw[leg] (T3) -- node[pos=1, xshift= 4pt, yshift=-4pt, font=\scriptsize] {$4$} ++( 0.9,-0.9); 

\draw[leg] (T4) -- node[pos=1, xshift= 4pt, yshift=4pt,  font=\scriptsize] {$1$} ++( 0.9, 0.9); 
\draw[leg] (T4) -- node[pos=1, xshift=-4pt, yshift=-4pt, font=\scriptsize] {$3$} ++(-0.9,-0.9); 
\draw[leg] (T4) -- node[pos=1, xshift= 4pt, yshift=-4pt, font=\scriptsize] {$2$} ++( 0.9,-0.9); 

\node[tensor] at (C)  {};
\node[tensor] at (T1) {};
\node[tensor] at (T2) {};
\node[tensor] at (T3) {};
\node[tensor] at (T4) {};

\end{tikzpicture}
        \end{adjustbox}
    \end{minipage}
    \hspace{10mm}
    \begin{minipage}{0.18\textwidth}
        \centering
        \begin{adjustbox}{clip,trim=0cm 0cm 0cm 0cm,width=\linewidth}
\begin{tikzpicture}[
  scale=0.5,
  every node/.style={transform shape},
  >=latex,
  tensor/.style={
    circle,
    draw=black!70,
    line width=0.9pt,
    minimum size=9mm,
    shade,
    top color=blue!15,
    bottom color=blue!55,
    preaction={fill=black,opacity=0.15,
      transform canvas={shift={(0.9pt,-0.9pt)}}}
  },
  leg/.style={line width=0.9pt}
]

\coordinate (T1) at (-0.8,  0.0);
\coordinate (T2) at ( 1.0, -1.6);
\coordinate (T3) at (-0.8, -3.2);
\coordinate (T4) at ( 1.0, -4.8);

\draw[leg] (T1) -- node[pos=0.5, yshift=8pt] {$1$} (T2);
\draw[leg] (T2) -- node[pos=0.5, yshift=8pt] {$3$} (T3);
\draw[leg] (T3) -- node[pos=0.5, yshift=8pt] {$2$} (T4);

\draw[leg] (T1) -- node[pos=1, xshift=-4pt, yshift=4pt, font=\scriptsize] {$3$} ++(-0.9, 0.9); 
\draw[leg] (T1) -- node[pos=1, xshift= 4pt, yshift=4pt, font=\scriptsize] {$2$} ++( 0.9, 0.9); 
\draw[leg] (T1) -- node[pos=1, xshift=-4pt, yshift=-4pt, font=\scriptsize] {$4$} ++(-0.9,-0.9); 

\draw[leg] (T2) -- node[pos=1, xshift= 4pt, yshift=4pt, font=\scriptsize] {$2$} ++( 0.9, 0.9); 
\draw[leg] (T2) -- node[pos=1, xshift= 4pt, yshift=-4pt, font=\scriptsize] {$4$} ++( 0.9,-0.9); 

\draw[leg] (T3) -- node[pos=1, xshift=-4pt, yshift=4pt, font=\scriptsize] {$1$} ++(-0.9, 0.9); 
\draw[leg] (T3) -- node[pos=1, xshift=-4pt, yshift=-4pt, font=\scriptsize] {$4$} ++(-0.9,-0.9); 

\draw[leg] (T4) -- node[pos=1, xshift= 4pt, yshift=4pt, font=\scriptsize] {$3$} ++( 0.9, 0.9); 
\draw[leg] (T4) -- node[pos=1, xshift= 4pt, yshift=-4pt, font=\scriptsize] {$4$} ++( 0.9,-0.9); 
\draw[leg] (T4) -- node[pos=1, xshift=-4pt, yshift=-4pt, font=\scriptsize] {$1$} ++(-0.9,-0.9); 

\node[tensor] at (T1) {};
\node[tensor] at (T2) {};
\node[tensor] at (T3) {};
\node[tensor] at (T4) {};

\end{tikzpicture}
        \end{adjustbox}
    \end{minipage}
    \caption{Examples of graphs corresponding to the second hierarchy for 
    $n=4$ particles. (Left). First non-trivial level $l=1$ of the finite Bethe lattice, where an initial tensor is contracted with $n$ copies. Subsequent levels of the Bethe lattice contract each of the tensors in the `outer' layer with $n-1$ other copies. (Right) Example tree tensor network of a path graph labeled by $\left[1, 3, 2\right]$, corresponding to the indices connecting the tensors from top to bottom. To calculate the associated bound for a given tree tensor network, we apply symmetric projectors on each of the open legs of a given type, and then contract the tensor with itself. Note that the symmetric projectors can be of different sizes; the tree tensor network on the right requires a symmetric projector on $2$ copies for the leg of type $1$, while requiring a symmetric projector on $4$ copies for the leg of type $4$.}
    \label{fig:example_trees}
\end{figure}
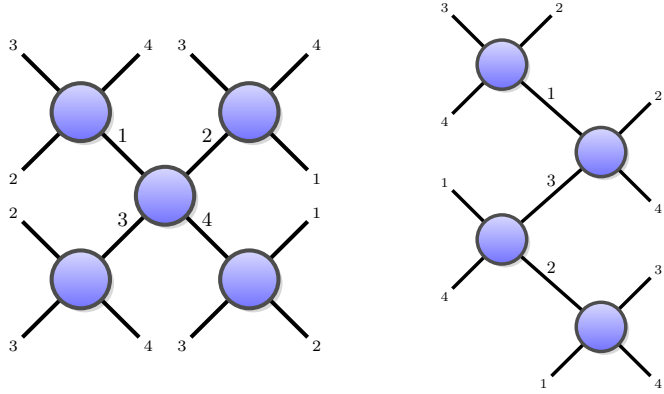

The first example shown is a finite Bethe lattice. The infinite Bethe lattice is the infinite $n$-regular tree, while the finite variant of level $l$ arises from the infinite one by choosing a root vertex and taking the induced subgraph on the vertices that are a distance at most $l$ from that root. The Bethe lattice is natural, since it does not favor any of the particles over the other. However, in practice we found that the bounds found through Bethe lattice TTNs were rather weak. Furthermore, the number of tensors grows exponentially in the parameter $l$, making runtime prohibitive for $l>2$.

In practice we found that the second example of path graph TTNs worked well in practice. However, unlike the Bethe lattice, for each level $k$ there are different ways to connect the labeled legs for a given path graph. A path graph TTN can then be specified by the sequence of legs that get connected. For example, the path graph TTN in Fig.~\ref{fig:example_trees} is specified by $\left[1, 3, 2\right]$, when starting from the top tensor. The different choices for fixed $k$ yielded similar results, but sequences of the form $\left[1, 2, \ldots, n, 1, 2,\ldots, n, 1, \ldots\right]$ worked well in practice.

For completeness, we show the bounds for different path graph TTNs. In particular, we consider superpositions of Dicke states with single and double excitations, i.e.~$\sqrt{s}\ket{D_{1}^{(5)}} + \sqrt{1-s}\ket{D_{2}^{(5)}}$. We observe here once more that the bounds are less tight for the more strongly entangled states.
\begin{figure}[h]
    \centering
    \includegraphics[width=0.5\linewidth]{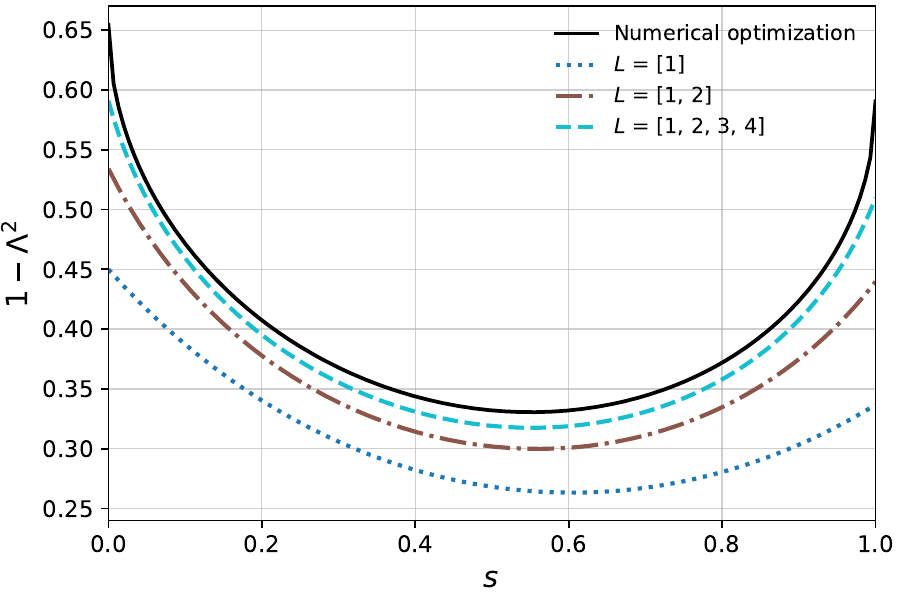}
\caption{Lower bounds on the geometric measure of superpositions of Dicke states of the form $\sqrt{s}\ket{D_{1}^{(5)}} + \sqrt{1-s}\ket{D_{2}^{(5)}}$. The lower bounds are found for three path graph TTNs parametrized by $L = \left[1\right]$, $L = \left[1, 2\right]$, $L = \left[1, 2, 3, 4\right]$. We also show a numerical upper bound, found by a brute-force search over a discretization of all symmetric product states.}
    \label{fig:Dicke_1vs2}
\end{figure}

\section{The second hierarchy 
$\mathfrak{H}_2$: proof of convergence}\label{app:proof_convergence_second_hierarchy}

A similar approach as to the first hierarchy can be used for the second hierarchy. In the following, we denote by $\ket{T}$ the tree tensor network (see Appendix \ref{app:trees}), by $k(i)$ the number of open legs of a given type $i$ in the tree $T$, by $\Pi_{k(i)}$ the projector onto the symmetric subspace of the legs of type $i$, and by $k$ the total number of vertices/`copies' in $T$. We then have the following.

\begin{theorem}
Given an $n$-partite state $\ket{\psi}$ and any associated $\ket{T}$ with $k$ vertices, and $k(i)$ open legs of each type $i$, we have the two-sided bound
\begin{equation}
    d_k \cdot ||\bigotimes_{i=1}^{n}\Pi_{k(i)} \ket{T}||^{2/k}
    \leq 
    \Lambda^2(\psi)
    \leq
    ||\bigotimes_{i=1}^{n}\Pi_{k(i)} \ket{T}||^{2/k},
\end{equation}
for any $k$. The coefficient $d_k:=\left(\prod_{i=1}^n \binom{k(i)+d-1}{k(i)}\right)^{-1/k}<1$ 
converges to one, i.e.~$\lim_{k\rightarrow\infty} d_k =1$. 
\end{theorem}
\begin{proof}
Let $\ket{a_1\ldots a_n}$ be the closest separable state to $\ket{\psi}$. Since by definition $\left|\braket{a_1\ldots a_n}{\psi}\right|=\Lambda$, we have that $\braket{\psi}{a_2\ldots a_n} = \Lambda \ket{a_1}$ up to some phase. Similar statements hold for the other $\ket{a_i}$ as well.

As such, we can rewrite the norm as

\begin{align}
\Lambda(\psi)^k&=\max_{a_1,\ldots, a_n \in \textrm{SEP}} \left|\bra{a_1}^{\otimes k(1)}\cdots \bra{a_n}^{\otimes k(n)}\ket{T}\right|\\
&=\max_{a_1,\ldots, a_n \in \textrm{SEP}} \left|\bra{a_1}^{\otimes k(1)}\cdots \bra{a_n}^{\otimes k(n)}\bigotimes_{i=1}^{n}\Pi_{k(i)}\ket{T}\right|\\
&\leq ||\bigotimes_{i=1}^{n}\Pi_{k(i)}\ket{T}||\ .
\end{align}
In the first step, we used that every leaf $v$ of $T$ has a unique leg that is not open. This implies that $\braket{a_1\ldots a_{j-1}a_{j+1}\ldots a_n}{T}=\Lambda\ket{T'}$ (up to a phase). Here $j$ is the type of the non-open leg of $v$, and $T'$ is the tree tensor network $T'$ but with $v$ deleted. By removing leaves recursively, the claim follows. In the last line we use that the maximization can be upper bounded by the norm. This proves the upper bound.

For the lower bound, we use that 
\begin{align}
||\bigotimes_{i=1}^{n}\Pi_{k(i)} \ket{T}||^2 &= \left(\prod_{i=1}^n \binom{k(i)+d-1}{k(i)}\right)\int_{\phi_1,...,\phi_n\in \mathrm{Haar}}\left \lVert\braket{\phi_1^{\otimes k(1)}\ldots\phi_n^{\otimes k(n)}}{T} \right \rVert^2 d\phi_1...d\phi_n \nonumber \\
&\leq  \left(\prod_{i=1}^n \binom{k(i)+d-1}{k(i)} \right)~\left \lVert\braket{a_1^{\otimes k(1)}\ldots a_n^{\otimes k(n)}}{T} \right \rVert^2
\nonumber\\
&=  \left(\prod_{i=1}^n \binom{k(i)+d-1}{k(i)}\right) ~\Lambda^{2k}(\psi)\ .
\end{align}

As before, it now remains to show that
\begin{align}
    \left(\prod_{i=1}^{n}\binom{k(i)+d-1}{k(i)}\right)^{1/k} &\leq \binom{k+d-1}{d-1}^{n/k}
\quad\substack{\rightarrow\\ k\rightarrow\infty}\quad 1\,.
\end{align}
where in the first step we used that $\binom{k(i)+d-1}{k(i)} = \binom{k(i)+d-1}{d-1} \leq \binom{k+d-1}{d-1}$ since $k(i)\leq k$, and in the second step we used the same argument as in Eq.~\eqref{eq:dk_convergence}.
\end{proof}

\section{Relation to the works by H.~A.~Helfgott and S.~Friedland}\label{app:connection_helfgott_friedland}

This present work is inspired by the work of S.~Friedland \cite{friedland2021upperboundsspectralnorm} which is itself 
a comment on a question by H.~A.~Helfgott \cite{helfgott}. So, 
in this Appendix, we want to explain this background and shortly 
show the connections to these works.

Interpreting a bipartite state $\ket{\psi}$ as a matrix $\psi_{ij}$ we 
find that $\max_{\ket{ab}}|\braket{\psi}{ab}|$ exactly amounts to the 
injective matrix norm $\norm{\psi}_{\rm inj}=\max_{a,b} |\sum_{ij}\psi_{ij}a_ib_j|$.
Note that for matrices this also equals the largest singular value or Ky-Fan norm. 
The definition directly generalizes to the tensor case, where the maximization depends on the chosen underlying field, which can be 
either real or complex. As an example consider the three-qubit state 
$\ket{B} = \frac{1}{2}(\ket{001}+\ket{010}+\ket{100}-\ket{111})$, which has $\norm{\psi}_{\rm inj, \mathbb{R}} = \frac{1}{2} < \frac{1}{\sqrt{2}} = \norm{\psi}_{\rm inj, \mathbb{C}}$ \cite{friedland2018nuclear}.

We now consider the real case first. In the case of real symmetric matrices 
it follows from the eigenvalue decomposition that 
\begin{align}
\norm{\psi}_{\rm inj}^k \leq \Tr[\psi^k].
\end{align}
This can be interpreted in a combinatorial way, as $\Tr[\psi^k]$ counts 
the weight of all closed walks of length $k$ in a weighted directed graph with adjacency 
matrix $\psi$, see for example Theorem 2.13 of \cite{chartrand2013first}. Note also that one 
can rewrite $\Tr[\psi^k]$ for symmetric $\psi$ using multiple copies of 
$\psi$ in the following way. For instance, for 
$k=3$ one has 
$\Tr[\psi^k]
= \sum_{ijrs} \psi_{ir}\psi_{rs}\psi_{sj}
= \sum_{kp} \big(\sum_{lmno} \delta_{ln} \delta_{mo} \psi_{kl} \psi_{mn} \psi_{op}
\big)
$, which bears some resemblance to the construction of the second hierarchy
in the main text.

In Ref.~\cite{friedland2021upperboundsspectralnorm} mainly real symmetric 
tensors were considered. For them, two hierarchies were constructed, 
in analogy to $\mathfrak{H}_1$ and $\mathfrak{H}_2$ in the main text. 
For $\mathfrak{H}_2$, only graphs where the underlying tree is a regular 
tree (also called Bethe lattice) were used. {These graphs can be constructed by taking a "root node", and connect it to $n$ nodes in the first layer. In the second layer one connects each of the $n$ nodes to $n-1$ new nodes, and repeats this scheme for every further layer, see also Fig.~\ref{fig:example_trees} and the description in Appendix~\ref{app:trees}}.  Note that in the case of 
symmetric tensors, the optimizing vectors in the injective norm can be 
assumed to be symmetric, i.e., $\norm{\psi}_{\rm inj}=\max_{a} \braket{\psi}{aa\dots}$ \cite{Hormander1954,hubener2009geometric}. 

It was then shown that in the real case the first hierarchy converges 
$\lim_{k\rightarrow \infty}\norm{\ket{F_k}}^{1/k} =: \rho_{1,\mathbb{R}}(\psi) \geq 
\norm{\psi}_{\rm inj,\mathbb{R}}$ for symmetric tensors, and the second hierarchy 
likewise converges to a finite value $\rho_2(\psi)\geq \norm{\psi}_{\rm inj,\mathbb{R}}$. 
In both cases equality holds for orthogonally diagonalizable tensors, 
i.e., tensors of the form 
\begin{align}
    \ket{\phi} = \sum_{i}\lambda_i \ket{a_ia_i\dots},
\end{align}
where the $\ket{a_i}$ form local orthonormal bases.

Our results now demonstrate that in the complex case (which was also 
shortly mentioned in \cite{friedland2021upperboundsspectralnorm}) 
$\rho_{1,\mathbb{C}}(\psi) = \rho_{2,\mathbb{C}}(\psi) = \norm{\psi}_{\rm inj,\mathbb{C}}$.

Additionally, our results also seem to indicate that the value $\rho_{1,\mathbb{R}}(\psi)$ is exactly given by the complex injective tensor norm $\norm{\psi}_{\rm inj,\mathbb{C}}$ for any real symmetric tensor $\psi$. This follows from the fact that in the first hierarchy we only compute the norm of the symmetrized tensor $\psi_{ijk}$, where in the case of a real tensor all coefficients are real, leading therefore to the same value in both the real and complex hierarchy. Since we show that the complex hierarchy converges exactly to the complex injective tensor norm, the same should hold for the real hierarchy. Thus, we find the seemingly strange phenomenon that a hierarchy defined on real symmetric tensors converges to their complex norm, i.e., 
$\lim_{k\rightarrow \infty}\norm{\ket{F_k}}^{1/k} = \rho_{1,\mathbb{R}}(\psi) = \norm{\psi}_{\rm inj,\mathbb{C}} \geq 
\norm{\psi}_{\rm inj,\mathbb{R}}$, where the inequality in the last step might be gapped for a given state, see the example of $\ket{B}$ above.

Finally, we would like to explain how our hierarchies can also be used to
find the injective tensor norm for the real case, which was the original problem
posed by H.~A.~Helfgott \cite{helfgott}. For that, we first note that the given
hierarchies allow to compute $\max_{\ket{abc}} |\bra{abc} X \ket{abc}|$ for complex
product vectors and complex, hermitean $X$, see also the main text. 

Now, consider a real three-particle state $\psi_{ijk}$. Using $X=\ketbra{\psi}{\psi}$ will just result in the complex injective tensor norm, so the trick is
to use a different $\tilde{X}$ for which the optimization gives the real injective
norm. Denoting $\rho=\ketbra{\psi}{\psi}=\rho^T$, we consider the real positive operator 
$\tilde X = (\rho + \rho^{T_A} + \rho^{T_B} + \rho^{T_C})/4$, which is invariant under every partial transposition. One therefore finds 
\begin{align}
    \max_{\ket{abc}} \bra{abc} \tilde X \ket{abc} = \max_{\ket{abc}} \Tr[\tilde X \ketbra{abc}{abc}]
    = \max_{\ket{abc}} \Tr[\tilde X \frac{\ketbra{a}{a}+\ketbra{a}{a}^T}{2} \otimes \frac{\ketbra{b}{b}+\ketbra{b}{b}^T}{2} \otimes \frac{\ketbra{c}{c}+\ketbra{c}{c}^T}{2} ],
\end{align}
where the last equality follows from the invariance of $\tilde X$ under each partial transposition. Naturally, $\rho_A=(\ketbra{a}{a}+\ketbra{a}{a}^T)/{2}=\rho_A^T = \rho_A^*$ is a real
single particle state, implying that the maximization above can also be understood as maximization over real product density matrices $\rho_A\otimes \rho_B \otimes \rho_C$. The eigenvalues and vectors of each real single particle state $\rho_X$
are real, too. So, due to convexity arguments, the maximum is attained for pure 
real product vectors, leading to 
\begin{align}
    \max_{\ket{abc}} \bra{abc} \tilde X \ket{abc} = \max_{\ket{abc}\text{ real}} \bra{abc} \tilde X \ket{abc}.
\end{align}
Consequently,
\begin{align}
    \norm{\psi}_{\rm inj,\mathbb{R}}^2 = \max_{\ket{abc}} \bra{abc} 
    \tilde X \ket{abc},
\end{align}
so the first and third hierarchies for $\tilde X$ converge to the real injective tensor norm of $\psi$.

\section{The third hierarchy
$\mathfrak{H}_3$: proof of convergence}\label{app:proof_convergence_third_hierarchy}

Let us first use the symmetric case to illustrate the idea. Let $\ket{\psi}$
be a symmetric state in $\mathcal{H}^{\otimes n}$ and
$\Lambda(\psi)=\max\abs{\braket{\psi}{{x^{\otimes n}}}}
=\max_{\rho\in \mathrm{SEP}\cap\mathrm{SYM}}\Tr(\ketbra{\psi}\rho)$,
where $\ket{x}\in\mathcal{H}$, where $\rho\in \mathrm{SEP}\cap\mathrm{SYM}$
means that $\rho$ is a separable state in the symmetric subspace.
Now, we take advantage of the quantum de Finetti theorem, or the so-called
symmetric extension, and approximate
$\rho$ with the reduced states of symmetric states $\rho_{(k+n)}$ in
$\mathcal{H}^{\otimes k+n}$.
Specifically, let $\xi_k$ be the solution of the following SDP
\begin{equation}
\begin{aligned}
    &\max_{\rho} \quad&&\Tr\qty[\qty(\id^{\otimes k}\otimes
    \ketbra{\psi}{\psi})\rho_{k+n}]\\
    &~\mathrm{s.t.} \quad && \Pi_{k+n} \rho_{k+n} \Pi_{k+n} = \rho_{k+n},\\
    & && \rho_{k+n}\ge 0,~\Tr(\rho_{k+n})=1,\\
\end{aligned}
\label{eq:SymmExt}
\end{equation}
then $\xi_k\ge\xi_{k+1}$ and
$\lim_{k\to\infty}\xi_k=\max\abs{\braket{\psi}{x^{\otimes n}}}$.
One can easily verify that Eq.~\eqref{eq:SymmExt} can be solved analytically as
\begin{equation}
\begin{aligned}
    \xi_k&=\lambda_{\max}\qty(\Pi_{k+n}
    \qty(\id^{\otimes k}\otimes\ketbra{\psi}{\psi})\Pi_{k+n})\\
    &=\lambda_{\max}\qty(\Pi_{k+n}
    \qty(\Pi_k\otimes\ketbra{\psi}{\psi})\Pi_{N+n}).
\end{aligned}
\label{eq:eigenbound}
\end{equation}

Generally, we consider the problem of optimizing
$\max_{\rho_{ABC}\in\mathrm{SEP}}\Tr(X_{ABC}\rho_{ABC})$,
where $X_{ABC}$ is some fixed Hermitian matrix.
In this case, we can still apply the symmetric extension but to each
subsystem individually, and the obtained hierarchy is still complete~\cite{doherty2005detecting}, i.e., the solutions $\xi_k$ of the
following SDPs give a complete hierarchy:
\begin{equation}
\begin{aligned}
    &\max_{\rho} \quad&&\Tr\qty[\qty(
    \id_A^{\otimes k-1}\otimes
    \id_B^{\otimes k-1}\otimes
    \id_C^{\otimes k-1}\otimes
    X_{ABC})\rho_k]\\
    &~\mathrm{s.t.} \quad &&
    \Pi_k^A\otimes\Pi_k^B\otimes\Pi_k^C \rho_k
    \Pi_k^A\otimes\Pi_k^B\otimes\Pi_k^C = \rho_k,\\
    & && \rho_k\ge 0,~\Tr(\rho_k)=1,\\
\end{aligned}
\label{eq:SymmExtGen}
\end{equation}
where $\rho_k\in\mathcal{H}_A^{\otimes k}\otimes
\mathcal{H}_B^{\otimes k}\otimes\mathcal{H}_C^{\otimes k}$ and
$\Pi_k^A$, $\Pi_k^B$, and $\Pi_k^C$ are symmetric projectors on 
$\mathcal{H}_A^{\otimes k}$, $\mathcal{H}_B^{\otimes k}$, and
$\mathcal{H}_C^{\otimes k}$, respectively.
Similarly, one can easily verify that Eq.~\eqref{eq:SymmExtGen} can be solved
analytically as
\begin{equation}
\begin{aligned}
    \xi_k&=\lambda_{\max}\qty[
    \Pi_k^A\otimes\Pi_k^B\otimes\Pi_k^C 
    \qty(\id_A^{\otimes k-1}\otimes
    \id_B^{\otimes k-1}\otimes
    \id_C^{\otimes k-1}\otimes
    X_{ABC})
    \Pi_k^A\otimes\Pi_k^B\otimes\Pi_k^C 
    ]\\
    &=\lambda_{\max}\qty[
    \Pi_k^A\otimes\Pi_k^B\otimes\Pi_k^C 
    \qty(\Pi_A^{\otimes k-1}\otimes
    \Pi_B^{\otimes k-1}\otimes
    \Pi_C^{\otimes k-1}\otimes
    X_{ABC})
    \Pi_k^A\otimes\Pi_k^B\otimes\Pi_k^C 
    ].
\end{aligned}
\label{eq:eigenboundGen}
\end{equation}

\section{The third hierarchy $\mathfrak{H}_3$: numerical approach}
\label{app:third_numerical}

To show how the third hierarchy can be efficiently calculated, we take the
qubit case ($\dim(\mathcal{H})=2$) as an example. Still, we first consider the
symmetric case. For any $n$-qubit symmetric pure state $\ket{\psi}$, we have
\begin{align}
\ket{\psi}&=\sum_{j=0}^na_j\ket{D_n^j}, \\
\Pi_k&=\sum_{\alpha=0}^k\ketbra{D_k^\alpha}{D_k^\alpha},\\
\Pi_k\otimes\ketbra{\psi}{\psi}
&=\sum_{\alpha=0}^k\ketbra{D_k^\alpha\psi}{D_k^\alpha\psi},
\end{align}
where $\ket{D_k^i}$ are the Dicke states. Then,
\begin{equation}
\begin{aligned}
    \Pi_{k+n}\ket{D_k^\alpha\psi}&=\sum_{j=0}^na_j\Pi_{k+n}\ket{D_k^\alpha D_n^j}\\
    &=\sum_{j=0}^na_j\sqrt{\frac{\binom{k}{\alpha}\binom{n}{j}}
    {\binom{k+n}{\alpha +j}}}\ket{D_{k+n}^{\alpha +j}}\\
    &=\sum_{\ell=0}^{k+n}a_{\ell-\alpha}\mu_{\alpha\ell}\ket{D_{k+n}^{\ell}},
\end{aligned}
\end{equation}
where
\begin{equation}
    \mu_{\alpha\ell}:=\sqrt{\frac{\binom{k}{\alpha}\binom{n}{\ell-\alpha}}
    {\binom{k+n}{\ell}}},
    \label{eq:defmu}
\end{equation}
and we have assumed that as $a_{\ell-\alpha}=\mu_{\alpha\ell}=0$
for $\alpha>\ell$ or $\alpha<\ell-n$. This also implies that the 
$\mu_{\alpha\ell}$ form a sparse matrix. Let
$M_k$ be the $(k+n+1)\times (k+n+1)$ matrix defined by
\begin{equation}
    [M_k]_{\ell m}=\sum_{\alpha=0}^k
    a_{\ell-\alpha}a^*_{m-\alpha}\mu_{\alpha\ell}\mu_{\alpha m}.
\end{equation}
Eq.~\eqref{eq:eigenbound} implies
that $\xi_k=\lambda_{\max}(M_k)$, or equivalently, $\xi_k=s^2_{\max}(V_k)$,
where $s_{\max}(\cdot)$ denotes the largest singular value and $V_k$ is
a $(k+n+1)\times (k+1)$ sparse matrix defined by
\begin{equation}
    [V_k]_{\ell \alpha}=a_{\ell-\alpha}\mu_{\alpha\ell}.
\end{equation}

Similarly, for a general $N$-qubit state
$\ket{\psi}=a_{i_1i_2\dots i_N}\ket{i_1i_2\dots i_N}$, we obtain the
$(k+1)^N\times (k+1)^N$
matrix $M_k$ 
\begin{equation}
    [M_k]_{\ell_1\dots\ell_N;m_1\dots m_N}=
    \sum_{\alpha_1,\dots,\alpha_N=0}^{k-1}
    a_{\ell_1-\alpha_1,\dots,\ell_N-\alpha_N}
    a_{m_1-\alpha_1,\dots,m_N-\alpha_N}^*
    \mu_{\alpha_1\ell_1}\mu_{\alpha_1m_1}
    \dots
    \mu_{\alpha_N\ell_N}\mu_{\alpha_Nm_N},
\end{equation}
where
\begin{equation}
    \mu_{\alpha\ell}=
    \begin{cases}
    \sqrt{\binom{k-1}{\alpha}/
    \binom{k}{\alpha}}=\sqrt{\frac{k-\alpha}{k}}\quad &\ell=\alpha,\\
    \sqrt{\binom{k-1}{\alpha}/
    \binom{k}{\alpha+1}}=\sqrt{\frac{\alpha+1}{k}}\quad &\ell=\alpha+1,\\
    0\quad &\text{others}.\\
    \end{cases}
\end{equation}
Then $\xi_k$ defined by
Eq.~\eqref{eq:SymmExtGen} is given by
$\xi_k=\lambda_{\max}(M_k)$,
or equivalently, $\xi_k=s^2_{\max}(V_k)$,
where $s_{\max}(\cdot)$ denotes the largest singular value and
$V_k$ is a $(k+1)^N\times k^N$ sparse matrix defined by
\begin{equation}
    [V_k]_{\ell_1\dots\ell_N;\alpha_1\dots \alpha_N}=
    a_{\ell_1-\alpha_1,\dots,\ell_N-\alpha_N}
    \mu_{\alpha_1\ell_1}
    \dots
    \mu_{\alpha_N\ell_N}.
\end{equation}

\section{Application to operators}\label{app:UPB}

The first and third hierarchies can also be formulated for the separable numerical range of an operator. As an example, we consider an unextendible product basis (UPB) on a two qutrit-state, given by~\cite{bennett1999unextendible},
\begin{align}
    \ket{\psi_0} &= \frac{1}{\sqrt{2}}\ket{0}(\ket{0}-\ket{1}), \\
    \ket{\psi_1} &= \frac{1}{\sqrt{2}}(\ket{0}-\ket{1})\ket{2}, \\
    \ket{\psi_2} &= \frac{1}{\sqrt{2}}\ket{2}(\ket{1}-\ket{2}), \\
    \ket{\psi_3} &= \frac{1}{\sqrt{2}}(\ket{1}-\ket{2})\ket{0} \text{  and} \\
    \ket{\psi_4} &= \frac{1}{3}(\ket{0}+\ket{1}+\ket{2}) (\ket{0}+\ket{1}+\ket{2}) .
\end{align}
An UPB consists of a maximal set of orthogonal product states, meaning that no further product state orthogonal to all UPB states exists. Considering the projector on the UPB $P_{\text{UPB}}= \sum_{i=0}^4 \ket{\psi_i}\bra{\psi_i}$, one finds then that the state $X_{\text{UPB}}\propto\id-P_{\text{UPB}}$ is entangled, since there exists no product state in its range. To construct an entanglement witness for this state one can consider $\mathcal{W}_{\text{UPB}} = P-\epsilon\id$, which should be positive on all product states \cite{guhne2003experimental,terhal2001family}, leading to
\begin{align}
    \bra{ab} W \ket{ab} = \bra{ab} P \ket{ab} - \epsilon \geq 0
    \Leftrightarrow  \bra{ab} P \ket{ab} \geq \epsilon \quad \forall \ket{ab}.
\end{align}
Thus, ideally one should choose $\epsilon = \min_{\ket{ab}} \bra{ab} P \ket{ab}$, or equivalently $\delta=1-\epsilon=\max_{\ket{ab}} \bra{ab} X_{\text{UPB}} \ket{ab} =M(X_{\text{UPB}})$. Formulating the first hierarchy we find therefore
\begin{align}
    \delta\leq \sqrt[k]{ \lambda_{max}\left( \Pi_k\otimes \Pi_k (X_{\text{UPB}})^{\otimes k} \Pi_k\otimes \Pi_k \right)}
\end{align}
while the third hierarchy in this case reads 
\begin{align}
    \delta\leq \lambda_{max}\left( \Pi_k\otimes \Pi_k X_{\text{UPB}}\otimes\id^{\otimes (k-1)} \Pi_k\otimes \Pi_k \right).
\end{align}

\section{Lower bounds on the convex roof for mixed states}\label{app:Mixed}

\begin{figure}[t]
\centering
\includegraphics[width=0.32\linewidth]{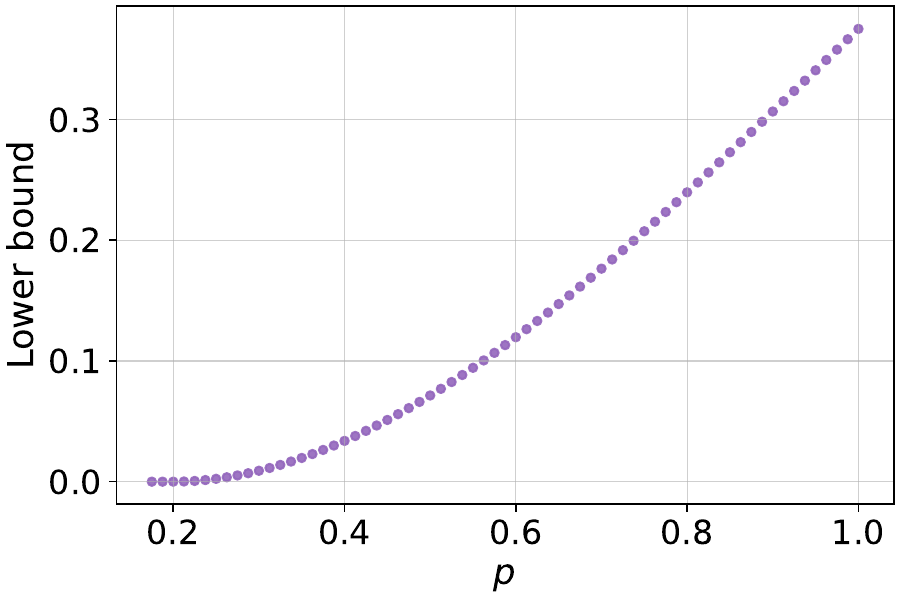}
\includegraphics[width=0.32\linewidth]{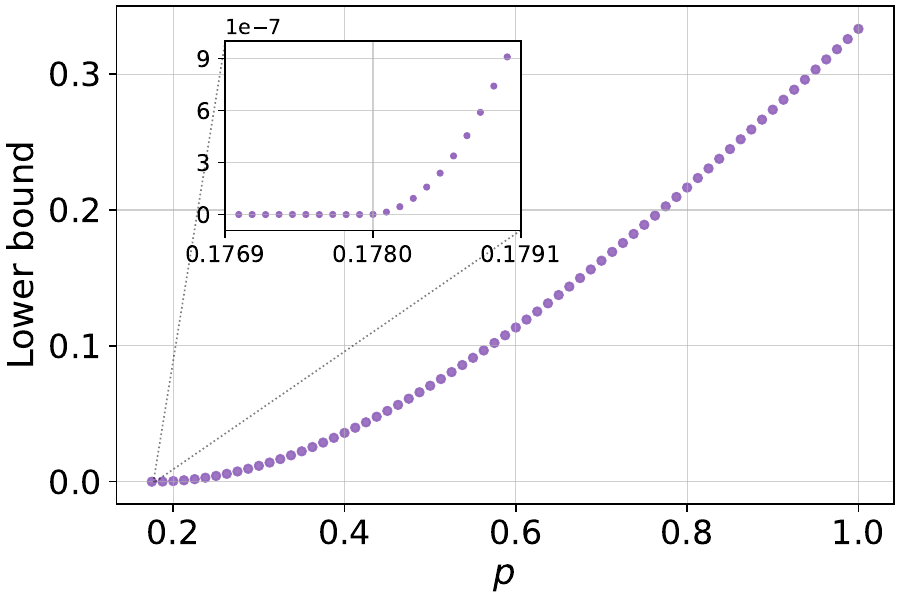}
\includegraphics[width=0.32\linewidth]{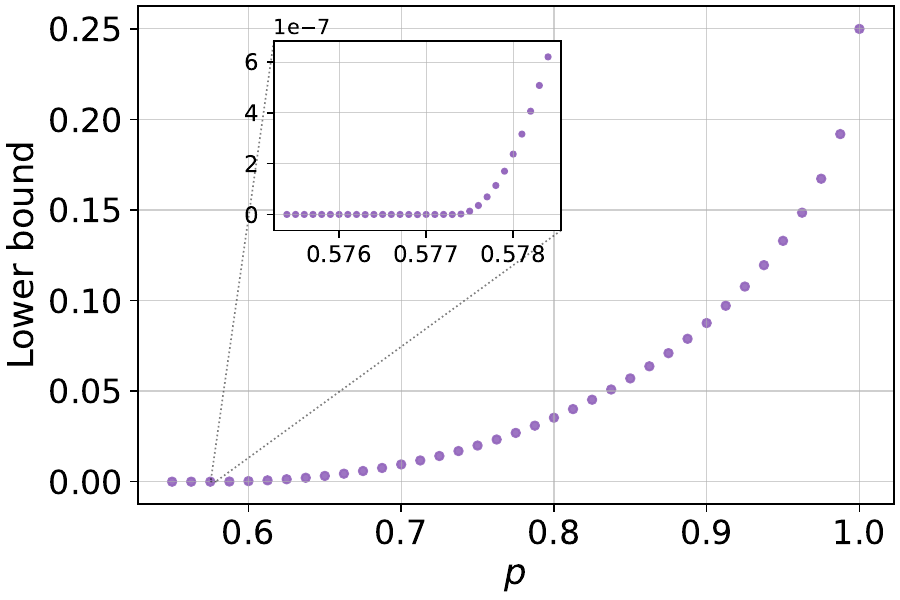}
\caption{(Left) Lower bounds for mixed states from the second level of 
the hierarchy (see Eq.~\eqref{eq-mixed-bound} in the main text) for GHZ states mixed 
with white noise, $\varrho(p) = p \ketbra{\mathrm{GHZ}} + (1-p) \openone/8$.  
(Middle) The estimate for W states with white noise,
$\varrho(p) = p \ketbra{W} + (1-p) \openone/8$.
These states are known to be fully separable for 
$p \leq 0.177$, entangled for $p>\sqrt{3}/(8+\sqrt{3}) \approx 0.1778$,
but still separable for any bipartition until $p \leq 0.2095$. The inset shows the solution of the SDP for the W state with higher precision, i.e., 
smaller duality gap, in the parameter regime close to the border of 
separability. (Right) Lower bounds for mixed states from the second level of 
the hierarchy (see Eq.~\eqref{eq-mixed-bound} in the main text) for Tao states mixed 
with white noise, $\varrho(p) = p \tau + (1-p) \openone/8$. The inset shows the solution of the SDP for the Tao state with higher precision, i.e., 
smaller duality gap, in the parameter regime close to the border of 
separability.
}
\label{fig:mixed}
\end{figure}

In order to test the estimation of the geometric measure for mixed states via 
Eq.~\eqref{eq-mixed-bound} in the main text we implemented the SDP using the solver Mosek. For three qubits, the SDP requires a computation time of about thirty seconds on a standard laptop. 

The results for the two examples discussed in the main text
are displayed in Fig.~\ref{fig:mixed}. On the one hand, one finds that the SDP 
for the first level underestimates the amount of entanglement for highly entangled 
pure states. On the other hand, it provides a simple and very strong entanglement 
test for weakly entangled states. This is in line with the observations made in 
the main text as well as in Appendix~\ref{app:trees}, where it was found that the 
hierarchies give good bounds for weakly entangled states. 

While the noisy GHZ and W states have frequently been discussed, the three-qubit
Tao state $\tau$ has been discovered  only recently \cite{ohst2024certifying}. It is 
given by
\begin{equation}
\tau = \frac{1}{4}
\begin{pmatrix}
a & 0 & 0 & 0 & 0 & 0 & 0 & -c
\\
0 & b & 0 & 0 & 0 & 0 & ic & 0
\\
0 & 0 & b & 0 & 0 & ic & 0 & 0
\\
0 & 0 & 0 & a & -c & 0 & 0 & 0
\\
0 & 0 & 0 & -c & b & 0 & 0 & 0
\\
0 & 0 & -ic & 0 & 0 & a & 0 & 0
\\
0 & -ic & 0 & 0 & 0 & 0 & a & 0
\\
-c & 0 & 0 & 0 & 0 & 0 & 0 & b
\end{pmatrix},
\end{equation}
with $a=\cos^2(\vartheta),$
$b=\sin^2(\vartheta),$
$c=\cos(\vartheta)\sin(\vartheta)$
and $\vartheta=\arccos(\sqrt{1/2+\sqrt{1/12}}).$ These states are biseparable 
for any fixed partition, an explicit decomposition is given in 
Ref.~\cite{ohst2024certifying}. On the other hand, the state is not fully separable
and hence entangled, actually the noisy state $\varrho(p) = p \tau + (1-p)\openone/8$
is provably entangled as long as $p \geq 0.5784$ and provably separable for 
$p \leq 0.5754.$ This is a large
noise robustness; actually, the state $\tau$ was found in Ref.~\cite{ohst2024certifying}
by an iteration of semidefinite programs looking for biseparable states with the highest
robustness of entanglement. Using the first level of the hierarchy, one directly finds
that the states $\varrho(p)$ are entangled for $p \geq 0.5775$, demonstrating the usefulness
of our approach to detect weakly entangled states.

\bibliography{references}

\end{document}